\documentclass[3p]{elsarticle}
\usepackage{amssymb,amsthm}
\usepackage{amsfonts}
\usepackage{relsize}
\usepackage{subfig}
\usepackage{mathtools}
\usepackage{color}
\usepackage[textsize=footnotesize,color=green!40]{todonotes}
\newtheorem{theorem}{Theorem}[section]

\newtheorem{conjecture}[theorem]{Conjecture}
\newtheorem{question}[theorem]{Question}
\newtheorem{proposition}[theorem]{Proposition}

\newtheorem{lemma}[theorem]{Lemma}
\newtheorem{observation}[theorem]{Observation}

\newtheorem{claim}{Claim}%[section]

\makeatletter
\@addtoreset{claim}{section}
\makeatother

\newcommand{\abs}[1]{\left|#1\right|}
\newcommand{\set}[1]{\left\{#1\right\}}
\newcommand{\paren}[1]{\left(#1\right)}
\newcommand{\um}{\scalebox{0.8}[1.0]{\(\, - \)}}

\newenvironment{proofclaim}[1][]%
    {\noindent \emph{Proof.} {}{#1}{}}{$~$\hfill $~\blacklozenge$ \vspace{0.2cm}}

\begin{document}

\begin{frontmatter}

\title{Between proper and strong edge-colorings of subcubic graphs}

\author[LaBRI]{Herve Hocquard\corref{cor1} \fnref{FRgrant}}
\ead{herve.hocquard@u-bordeaux.fr}
%\address{AGH University of Science and Technology, al. A. Mickiewicza 30, 30-059 Krakow, Poland}

\author[LaBRI]{Dimitri Lajou\fnref{FRgrant}} %,fn2
\ead{dimitri.lajou@u-bordeaux.fr}

\author[FIS]{Borut Lu\v{z}ar\fnref{FRgrant}} %,fn2
\ead{borut.luzar@gmail.com}

\cortext[cor1]{Corresponding author}
%\fntext[FRgrant]{}

\address[LaBRI]{Univ. Bordeaux, CNRS, Bordeaux INP, LaBRI, UMR 5800, F-33400, Talence, France}
\address[FIS]{Faculty of Information Studies, Novo mesto, Slovenia}

\begin{abstract}
	In a proper edge-coloring the edges of every color form a matching.
	A matching is \textit{induced} if the end-vertices of its edges induce a matching.
	A \textit{strong edge-coloring} is an edge-coloring in which the edges of every color form an induced matching.
	We consider intermediate types of edge-colorings, where edges of some colors are allowed to form matchings, 
	and the remaining form induced matchings. 
	Our research is motivated by the conjecture proposed in a recent paper of Gastineau and Togni on $S$-packing edge-colorings 
	(On $S$-packing edge-colorings of cubic graphs, Discrete Appl. Math. 259 (2019), 63--75)
	asserting that by allowing three additional induced matchings, one is able to save one matching color.
	We prove that every graph with maximum degree $3$ can be 
	decomposed into one matching and at most $8$ induced matchings,
	and two matchings and at most $5$ induced matchings.
	We also show that if a graph is in class I, 
	the number of induced matchings can be decreased by one, 
	hence confirming the above-mentioned conjecture for class I graphs.
\end{abstract}


\end{frontmatter}

%\section{TODO's}
%\begin{itemize}
%	\item{} Check about the induced colorings. Shall we use $\pi'$ also, or just $\pi$ for both graphs.
%\end{itemize}

\section{Introduction}
\label{sec:intro}

A \textit{proper edge-coloring} of a graph $G=(V,E)$ is 
an assignment of colors to the edges of $G$ such that adjacent edges are colored with distinct colors.
Due to a remarkable result of Vizing~\cite{Viz64}, 
we know that the minimum number of colors needed to color the edges of a graph $G$, 
the \textit{chromatic index} of $G$ (denoted by $\chi'(G)$),
is either $\Delta(G)$ or $\Delta(G) + 1$, $\Delta(G)$ being the maximum degree of $G$.
The graphs with the former value of the chromatic index are commonly said to be in \textit{class~I},
and the latter in \textit{class~II}.

In this paper, we are interested in graphs with maximum degree $3$, to which we will refer as \textit{subcubic graphs}.
We need at most $4$ colors to color such graphs; the complete graph on four vertices with one edge subdivided being
the smallest representative of a class~II subcubic graph, and the Petersen graph being the smallest $2$-connected class II cubic graph.
For subcubic graphs of class~II, it has been shown that they can be colored in such a way 
that one of the colors (usually denoted $\delta$) is used relatively rarely (cf.~\cite{AlbHaa96,FouVan13}).
This motivates the question if the edges of color $\delta$ can be pairwise distant.
Note that we consider the distance between edges as the distance between the corresponding vertices in the line graph,
i.e., adjacent edges are said to be at distance $1$.
Payan~\cite{Pay77} and independently Fouquet and Vanherpe~\cite{FouVan13} proved that 
every subcubic graph with chromatic index $4$ admits a proper edge-coloring
such that the edges of one color are at distance at least $3$,
i.e., the end-vertices of those edges induce a matching in the graph.

Gastineau and Togni~\cite{GasTog19} investigated a generalization of edge-colorings with the property described above.
For a given non-decreasing sequence of integers $S=(s_1,\dots,s_k)$, 
an $S$-packing edge-coloring of a graph is a decomposition of edges into disjoint sets $X_1,\dots,X_k$, 
where the edges in the set $X_i$ are pairwise at distance at least $s_i+1$.
A set $X_i$ is called an \textit{$s_i$-packing}; 
a $1$-packing is simply a matching, and a $2$-packing is an induced matching.
To simplify the notation, we denote repetitions of same elements in $S$ using exponents, 
e.g., $(1,2,2,2)$ can be written as $(1,2^3)$.

The notion of $S$-packing edge-colorings is motivated by its vertex counterpart, 
introduced by Goddard and Xu~\cite{GodXu12} as a natural generalization of the packing chromatic number~\cite{GodHedHedHarRal08}.
In~\cite{GasTog19}, the authors consider $S$-packing edge-colorings of subcubic graphs with prescribed number of $1$'s in the sequence.
Vizing's result translated to $S$-packing edge-coloring gives that every subcubic graph admits a $(1,1,1,1)$-packing edge-coloring,
while class I subcubic graphs are $(1,1,1)$-packing edge-colorable. 
Moreover, by Payan's, Fouquet's and Vanherpe's result, we have that there is a $(1,1,1,2)$-packing edge-coloring for any subcubic graph.
\begin{theorem}[Payan~\cite{Pay77}, and Fouquet \& Vanherpe~\cite{FouVan13}]
	\label{thm:1112}
	Every subcubic graph admits a $(1,1,1,2)$-packing edge-coloring.
\end{theorem}
Here $2$ cannot be changed to $3$, due to the Petersen and the Tietze graphs (depicted in Fig.~\ref{fig:pet}): 
they both have chromatic index $4$, 
and we need at least two edges of each color. 
Since every pair of edges is at distance at most $3$, we have the tightness.
However, Gastineau and Togni do believe that the following conjecture is true.
\begin{conjecture}[Gastineau and Togni~\cite{GasTog19}]
	Every cubic graph different from the Petersen and the Tietze graph is $(1,1,1,3)$-packing edge-colorable.
\end{conjecture}

\begin{figure}[ht]
	$$
		\includegraphics{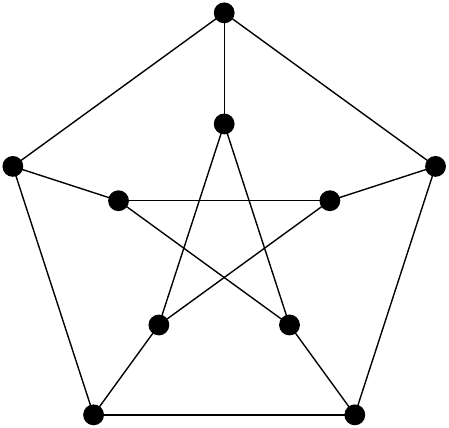} \hspace{2cm}
		\includegraphics{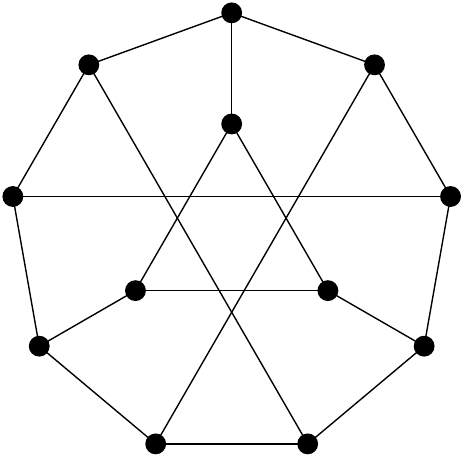}
	$$
	\caption{The Petersen (left) and the Tietze graph (right) admit a $(1,1,1,2)$-packing edge-coloring, 
		and $2$ cannot be increased to $3$.}
	\label{fig:pet}
\end{figure}

Clearly, reducing the number of $1$'s in sequences increases the total number of needed colors, i.e., the length of the sequence.
In fact, if there is no $1$ in a sequence, i.e., the edges of every color class induce a matching, 
the coloring is called a \textit{strong edge-coloring}. 
It has been proved by Andersen~\cite{And92} and independently by Hor\'{a}k, Qing, and Trotter~\cite{HorHeTro93}
that every subcubic graph admits a strong edge-coloring with at most $10$ colors, 
i.e., a $(2^{10})$-packing edge-coloring.
The number of colors is tight, e.g., the Wagner graph in Fig.~\ref{fig:strong} needs $10$ colors for a strong edge-coloring.
\begin{figure}[ht]
	$$
		\includegraphics{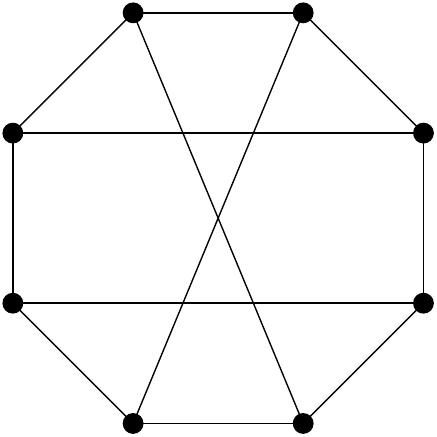}
	$$
	\caption{The Wagner graph is the smallest cubic graph which needs $10$ colors for a strong edge-coloring.}
	\label{fig:strong}
\end{figure}
Let us remark here that the Wagner graph is in class I, meaning that smallest chromatic index
does not necessarily mean less number of colors for a strong edge-coloring of a graph.

Proper and strong edge-coloring of subcubic graphs have been studied extensively already in the previous century.
In~\cite{GasTog19}, Gastineau and Togni started filling the gap by considering $(1^k,2^{\ell})$-packing edge-colorings for $k \in \set{1,2}$.
They proved that every cubic graph with a $2$-factor admits a $(1,1,2^5)$-packing edge-coloring, and the number of required $2$-packings 
reduces by one if the graph is class I. 
For the case with one $1$-packing, 
they remark that using the bound for the strong edge-coloring
one obtains that every subcubic graph admits a $(1,2^9)$-packing edge-coloring.
These bounds are clearly not tight, and they propose a conjecture (the items $(a)$ and $(c)$ in Conjecture~\ref{conj:main}), 
which motivated the research presented in this paper. The case $(b)$ has been formulated as a question, and we added the case $(d)$,
due to af\mbox{}f\mbox{}irmative results of computer tests on subcubic graphs of small orders.
\begin{conjecture}
	\label{conj:main}
	Every subcubic graph $G$ admits:
	\begin{itemize}
		\item[$(a)$] a $(1,1,2^4)$-packing edge-coloring~\cite{GasTog19};
		\item[$(b)$] a $(1,2^7)$-packing edge-coloring~\cite{GasTog19};
		\item[$(c)$] a $(1,1,2^3)$-packing edge-coloring if $G$ is in class I~\cite{GasTog19};
		\item[$(d)$] a $(1,2^6)$-packing edge-coloring if $G$ is in class I.
	\end{itemize}
\end{conjecture}

The conjectured bounds, if true, are tight. 
For the cases $(a)$ and $(b)$, a subcubic graph that achieves the upper bound is the complete bipartite graph $K_{3,3}$ with one subdivided edge 
(the left graph in Fig.~\ref{fig:112}).
Recall that this graph is also class II and needs $10$ colors for a strong edge-coloring, 
%hence achieving the upper bound for all four types of colorings considered in this paper.
hence achieving the upper bounds for proper edge-coloring, strong edge-coloring, $(1,1,2^4)$-packing edge-coloring, and $(1,2^7)$-packing edge-coloring.
%For each $1$-packing, we have at most three edges, and there remain $4$ and $7$, respectively, to be in a separate $2$-packing each.
Indeed, each $1$-packing contains at most three edges and each $2$-packing contains at most one edge.
An analogous argument holds for the cases $(c)$ and $(d)$ on the complete bipartite graph $K_{3,3}$.
\begin{figure}[ht]
	$$
		\includegraphics{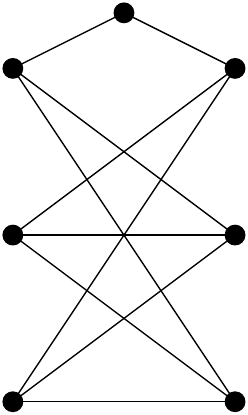} \hspace{2cm}
		\includegraphics{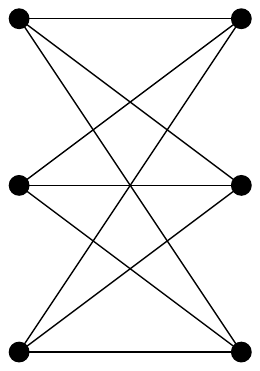}
	$$
	\caption{The smallest subcubic graph which does not admit a $(1,1,2^3)$-packing edge-coloring nor a $(1,2^6)$-packing edge-coloring (left),
		and the smallest class I subcubic graph which does not admit a $(1,1,2^2)$-packing edge-coloring nor a $(1,2^5)$-packing edge-coloring (right).}
	\label{fig:112}
\end{figure}
%\cmt{We do not have a $2$-connected example different from the above!}

Conjecture~\ref{conj:main} bridges two of the most important edge-colorings, proper and strong,
basically claiming that each $1$-packing could be replaced by three $2$-packings. Indeed, if such operations were possible, 
then one could transform a $(1,1,1,2)$-packing edge-coloring into a strong edge-coloring. 
Note that this does not apply to subclasses of graphs, 
e.g., the Wagner graph needs $10$ colors for a strong edge-coloring and it is in class I.
%\cmt{In Conclusion: However, since the methods used for strong do not apply...}

This paper contributes to answering the conjecture by providing upper bounds with one additional color for all four cases of Conjecture~\ref{conj:main}.
\begin{theorem}
	\label{thm:main}
	Every subcubic graph $G$ admits:
	\begin{itemize}
		\item[$(a)$] a $(1,1,2^5)$-packing edge-coloring;
		\item[$(b)$] a $(1,2^8)$-packing edge-coloring;
		\item[$(c)$] a $(1,1,2^4)$-packing edge-coloring if $G$ is in class I;
		\item[$(d)$] a $(1,2^7)$-packing edge-coloring if $G$ is in class I.
	\end{itemize}
\end{theorem}

The structure of the paper is the following. 
We begin by presenting notation, definitions and auxiliary results in Section~\ref{sec:prel}.
In Section~\ref{sec:11}, we give proofs of the cases $(a)$ and $(c)$ of Theorem~\ref{thm:main}.
In Sections~\ref{sec:1all} and~\ref{sec:1classI}, we proof the cases $(b)$ and $(d)$ of Theorem~\ref{thm:main} 
in even stronger settings. 
We conclude the paper with an overview of open problems and possible further work on this topic.

\section{Preliminaries}
\label{sec:prel}
	
We call a vertex of degree $k$, at most $k$, and at least $k$ a \textit{$k$-vertex}, a \textit{$k^-$-vertex}, and a \textit{$k^+$-vertex}, respectively.
We denote the graph obtained from a graph $G$ by removing a set of vertices $X$ as $G \setminus X$.
When $X = \set{v}$ is a singleton, we simply write $G - v$.
An analogous notation is used for sets of edges.

As usual, the set of vertices adjacent to a vertex $v$ is denoted $N(v)$, and called the \textit{neighborhood of $v$}.
For a vertex $v$, we denote the set of edges incident to $v$ by $N'(v)$,
and the edges incident to the neighbors of $v$ (including the edges in $N'(v)$) by $N''(v)$.
We refer to the former as the \textit{edge-neighborhood of $v$} and to the latter as the \textit{$2$-edge-neighborhood of $v$}.
Analogously, we define the edge-neighborhood and the $2$-edge-neighborhood of an edge $e$.

When coloring the edges, we deal with two types of colors. 
The ones allowing the edges of those colors to be at distance at least $2$ we call the \textit{$1$-colors}, 
and the one requiring the edges to be at distance at least $3$ are called the \textit{$2$-colors}.
An edge colored with a $1$-color (resp. a $2$-color) is a \textit{$1$-edge} (resp. a \textit{$2$-edge}).
We denote the number of $1$-edges (resp. $2$-edges) incident with a vertex $v$ by $c_1(v)$ (resp. $c_2(v)$).
For an edge $uv$, we denote by $A_2(uv)$ the number of available $2$-colors, 
i.e., the $2$-colors with which the edge can be colored without violating the coloring assumptions.
 
%In our proofs, we will often use the following simple observation. 
%\begin{observation}\todo{ Do we keep this observation ?}
%	\label{obs:core}
%	If an edge $e$ in a graph $G$ has at least one available color after all the edges in its $2$-edge-neighborhood are colored, 
%	we may ignore it. %We call such an edge {\em free}.
	%we may ignore it, i.e., we only consider the $S$-packing edge-coloring of $G - e$. We call such an edge {\em free}.
	%After a recursive removal of all free edges from $G$, we obtain a graph, which we call the {\em core} of $G$ and we denote it by $\zeta(G)$. 
	%Observe that $G$ is $S$-packing edge-colorable if and only if $\zeta(G)$ is $S$-packing edge-colorable. 
	%Hence, also every graph $G$ whose core is the null graph, is $S$-packing edge-colorable.
%\end{observation}

Sometimes, we will need a more careful analysis of choosing colors from the lists of available colors.
For that purpose, we will use the classical result due to Hall~\cite{Hal35}.
\begin{theorem}[Hall's Theorem~\cite{Hal35}]
	\label{th:Hall} 
	Let $\mathcal{A}=(A_i)_{i \in I}$ be a finite family of (not necessarily distinct) subsets of a finite set $A$. 
	A system of representatives for the family $\mathcal{A}$ is a set $\{a_i, i \in I\}$ 
	of distinct elements of $A$ such that $a_i \in A_i$ for all $i \in I$. 
	$\mathcal{A}$ has a system of representatives if and only if $|\bigcup_{i \in J} A_i| \ge |J|$ for all subsets $J$ of $I$.
\end{theorem}

Perhaps the strongest tool for determining if one can always choose colors from the lists of 
available colors such that given conditions are satisfied is the following result, due to Alon~\cite{Alo99}.
\begin{theorem}[Combinatorial Nullstellensatz~\cite{Alo99}]
	\label{thm:null} 
	Let $\mathbb{F}$ be an arbitrary field, 
	and let $P=P(X_1,\ldots,X_n)$ be a polynomial in $\mathbb{F}[X_1,\ldots,X_n]$.
	Suppose that the coefficient of a monomial $X_1^{k_1}\ldots X_n^{k_n}$, 
	where each $k_i$ is a non-negative integer, 
	is non-zero in $P$ 
	and the degree ${\rm deg}(P)$ of $P$ equals $\sum_{i=1}^n k_i$.
	If moreover $S_1,\ldots,S_n$ are any subsets of $\mathbb{F}$ with $|S_i|>k_i$ for $i=1,\ldots,n$,
	then there are $s_1\in S_1,\ldots,s_n\in S_n$ so that $P(s_1,\ldots,s_n) \neq 0$.
\end{theorem}
In short, $P$ being the chromatic polynomial of a graph $G$, 
if there is a monomial (of proper degree) of $P$ with non-zero coefficient, then there exists a coloring of $G$.

When considering lists of available colors for an edge, 
we are in fact dealing with the list version of a coloring. 
We say that $L$ is an \textit{edge-list-assignment} for a graph $G$ if it assigns a list $L(e)$ of possible colors to each edge $e$ of $G$.
If $G$ admits a strong edge-coloring $\sigma$ such that $\sigma(e) \in L(e)$ for all edges in $E(G)$, 
then we say that $G$ is \textit{strong $L$-edge-colorable} or $\sigma$ is a {\em strong $L$-edge-coloring} of $G$.
The graph $G$ is \textit{strong $k$-edge-choosable} if it is strong $L$-edge-colorable for every edge-list-assignment $L$, 
where $|L(e)| \ge k$ for every $e \in E(G)$. 
The \textit{list strong chromatic index} $\chi_{\textrm{ls}}'(G)$ of $G$ is the minimum $k$ such that $G$ is strong $k$-edge-choosable. 

We will use the following result, due to Zhang, Liu, and Wang~\cite{ZhaLiuWan02}
which established a result on an adjacent vertex-distinguishing list edge-coloring of cycles,
i.e., proper list edge-coloring where the sets of colors for every pair of adjacent vertices are distinct.
It is easy to see that such a coloring is also a strong edge-coloring of a cycle, and we write the statement in this language.
\begin{theorem}[Zhang, Liu \& Wang~\cite{ZhaLiuWan02}]
	\label{thm:listcyc}
	Let $n$ be an integer with $n \ge 3$. Then, 
	\begin{itemize}
		\item[$(i)$] $\chi_{\textrm{ls}}'(C_n) = 5$ if $n=5$;
		\item[$(ii)$] $\chi_{\textrm{ls}}'(C_n) = 4$ if $n \not\equiv 0 \bmod{3}$;
		\item[$(iii)$] $\chi_{\textrm{ls}}'(C_n) = 3$ if $n \equiv 0 \bmod{3}$.
	\end{itemize}
\end{theorem}

%
%
%  (1, 1, 2^k)-coloring
%
%
\section{Proofs of the cases $(a)$ and $(c)$ of Theorem~\ref{thm:main}}
\label{sec:11}

We begin with the cases of Theorem~\ref{thm:main} using two $1$-colors.
These two cases provide straightforward extensions of the results by Gastineau and Togni~\cite{GasTog19},
who established them for bridgeless cubic graphs.
The extension comes from the following easy observation.

\begin{proposition}
	\label{prop:K5}
	Let $G$ be a subcubic graph 
	and let $X$ be a set of edges in $G$ such that every two edges in $X$ are at distance exactly $2$.
	Then, $X$ contains at most $5$ edges.
	Moreover, if $|X| = 5$, then $G$ is cubic with $10$ vertices.
\end{proposition}

\begin{proof}
	Let $X$ be a set of edges in a subcubic graph $G$ satisfying assumptions of the proposition.
	For an edge $e \in X$, we have at most four adjacent edges, say $e_1$, $e_2$, $e_3$, and $e_4$.
	Each edge $e_i$, $1 \le i \le 4$, can be adjacent to at most one other edge from $X$,
	since otherwise there would be two edges at distance $1$ in $X$.
	This means, $X$ contains at most $5$ edges.
	
	In the case when $|X| = 5$, every edge of $G$ not in $X$ connects two edges of $X$, 
	hence every vertex of $G$ is an end-vertex of some edge from $X$ and thus the number 
	of vertices in $G$ is 10. Since every edge from $X$ is adjacent to four edges, we infer that $G$ is cubic.
\end{proof}

Now, we are ready to prove the cases $(a)$ and $(c)$ of Theorem~\ref{thm:main}.
\begin{proof}[Proof of Theorem~\ref{thm:main}$(a)$ and $(c)$]
	We begin with the case $(a)$.
	Let $G$ be a subcubic graph (we may assume it is connected) 
	and let $\pi$ be a $(1,1,1,2)$-packing edge-coloring of $G$
	(which exists by Theorem~\ref{thm:1112}).
	To establish the statement, we only need to replace one $1$-color in $\pi$ with four $2$-colors.
	Let $X$ be a set of all the edges in $G$ colored by one $1$-color in $\pi$, 
	and let $G^*$ be the graph obtained from $G$ by contracting all the edges in $X$ (and removing loops that are created in the process).
	Clearly, $G^*$ has maximum degree at most $4$, and it is $4$-vertex-colorable by the Brooks' Theorem, 
	unless it is isomorphic to $K_5$. 
	Observe that vertex coloring of $G^*$ induces a strong edge-coloring of the edges in $X$.
	Furthermore, by Proposition~\ref{prop:K5}, the only graphs in which it may happen that five colors are needed
	to color $G^*$, are cubic with $10$ vertices. 
	For these graphs we have even determined that they even admit a $(1,1,2^4)$-packing edge-coloring, 
	and thus establish the case $(a)$.
	
	The case $(b)$ follows immediately from the argument above, since we do not have an extra $2$-color in the coloring $\pi$.
\end{proof}

%
%
%  (1, 2^8)-coloring
%
%
\section{Proof of the case $(b)$ of Theorem~\ref{thm:main}}
\label{sec:1all}

In order to prove Theorem~\ref{thm:main}$(b)$, we prove a bit stronger result.
We say that a $(1,2^8)$-packing edge-coloring of a subcubic graph $G$ with the color set $\set{0,1,\dots,8}$, 
where $0$ is a $1$-color and the others are $2$-colors, 
is a \textit{good $(1,2^8)$-packing edge-coloring} if no $2^-$-vertex of $G$ is incident with a $1$-edge (i.e., an edge colored $0$).

\begin{theorem}
	\label{thm:128_good}
	Every subcubic graph admits a good $(1,2^8)$-packing edge-coloring.
\end{theorem}

\begin{proof}
	We prove Theorem~\ref{thm:128_good} by contradiction. 
	Let $G$ be a minimal counterexample to the theorem in terms of $\abs{V(G)} + \abs{E(G)}$.
	Clearly, $G$ is connected and has maximum degree $3$.
	In the following claims, we establish some structural properties of $G$
	which will eventually yield a contradiction on the existence of $G$.
	In most of the claims, we consider a graph $G'$ smaller than $G$, 
	which, by the minimality of $G$, admits a good $(1,2^8)$-packing edge-coloring $\pi$,
	and we show that $\pi$ can be extended to $G$ by recoloring some edges of $G'$ 
	and coloring the edges of $G$ not being colored by $\pi$.

\begin{claim}
	\label{128_simple}
	$G$ is simple.
\end{claim}

\begin{proofclaim}
	Suppose there are vertices $u$ and $v$ in $G$ connected by at least two parallel edges.
	Remove one of the edges (call it $e$) between them to obtain a smaller graph $G'$, 
	and let $\pi$ be a good $(1,2^8)$-packing edge-coloring of $G'$. 
	We can extend $\pi$ to $G$, since $A_2(e) \le 7$, and hence there is an available $2$-color for $e$.
	
	Loops can be removed in a similar fashion.
\end{proofclaim}

\begin{claim}
	\label{128_deg2-}
	$G$ is cubic.
\end{claim}

\begin{proofclaim}
	Suppose first that there exists a $1$-vertex $v$ adjacent to a vertex $u$ in $G$. 
	By the minimality of $G$, there exists a good $(1,2^8)$-packing edge-coloring $\pi$ of $G'= G - v$,
	meaning $u$ is not incident with a $1$-edge. 
	We can extend $\pi$ to $G$, 
	by coloring $uv$ with any of the (at least two) $2$-colors that do not appear in the $2$-edge-neighborhood of $u$.
	Hence, $G$ does not contain $1$-vertices.

	Suppose now that there exists a $2$-vertex $v$ adjacent to the vertices $u$ and $w$. 
	By the minimality of $G$, $G' = G - v$ admits a good $(1,2^8)$-edge-coloring $\pi$, 
	and hence $u$ and $w$ are not incident with a $1$-edge.
	We show that $\pi$ can be extended to $G$ as follows.
	First observe that if
	$A_2(uv) \ge 2$ and $A_2(vw) \ge 1$,
	or
	$A_2(uv) \ge 1$ and $A_2(vw) \ge 2$,	
	then $\pi$ can be extended and we are done.
	So, we may assume that $A_2(uv) \le 1$ and $A_2(vw) \le 1$.
	It follows that $u$ and $w$ are both $3$-vertices in $G$, and moreover, $u$ and $w$ are not adjacent.	
	Next, let $u_1$ and $u_2$ be the neighbors of $u$ distinct from $v$, and analogously, 
	let $w_1$ and $w_2$ be the neighbors of $w$.
	By the above argument, $d_2(u_1) + d_2(u_2) \ge 5$, 
	meaning that at least one of $u_1$ and $u_2$ is a $3$-vertex not incident with a $1$-edge, say $u_1$.
	Moreover, $u_1$ is not adjacent to $w$.
	Now, we recolor $uu_1$ with $0$. 
	By an analogous argument, we have that we can recolor one of the edges adjacent to $w$ with $0$,
	obtaining a contradiction on the number of available colors.
	Hence, $\pi$ can be extended to $G$.
\end{proofclaim}

\begin{claim}\label{128_2con}
	$G$ is $2$-connected.
\end{claim}

\begin{proofclaim}
	Since $G$ is cubic, the claim is equivalent to saying that $G$ is bridgeless.
	Suppose the contrary and let $uv$ be a bridge in $G$.
	Let $G_u$ (resp. $G_v$) be the component of $G - uv$ containing $u$ (resp. $v$).
	By the minimality of $G$, there is a good $(1,2^8)$-packing edge-coloring $\pi_u$ of $G_u + uv$
	and a good $(1,2^8)$-packing edge-coloring $\pi_v$ of $G_v + uv$. 
	Meaning that the edge $uv$ is in both cases colored with a $2$-color. 
	Now, we permute the $2$-colors in the coloring $\pi_v$ so that the color of $uv$ 
	is the same in both colorings, $\pi_u$ and $\pi_v$, and that the colors on the other two edges incident with $v$ 
	are distinct from the colors on the other two edges incident with $u$ (except possibly the color $0$).
	In this way, we obtain a good $(1,2^8)$-packing edge-coloring of $G$, a contradiction.
\end{proofclaim}

\begin{claim}\label{128_girth4}
	$G$ does not contain $3$-cycles.
\end{claim}

\begin{proofclaim}
	Suppose the contrary and let $C = uvw$ be a $3$-cycle in $G$.
	If $C$ is adjacent to two other $3$-cycles, then $G$ is the complete graph on four vertices, and hence colorable.
	So, we may assume $C$ is adjacent to at most one $3$-cycle. We consider two cases.
	
	Suppose first that $C$ is adjacent to a $3$-cycle $C' = uvx$.
	Let $w'$ and $x'$ be the neighbors of $w$ and $x$, respectively, distinct from $u$ and $v$.
	Since $G$ is bridgeless, by Claim~\ref{128_2con}, $w' \neq x'$.
	Now, by the minimality of $G$, there is a coloring $\pi$ of $(G \setminus \set{uvwx}) + w'x'$ (if $w'x'$ are already connected, we add a parallel edge).
	Let $\pi(w'x') = a$ and consider the coloring of $G$ induced by $\pi$ by coloring $xx'$ and $ww'$ by $a$.
	The edges $uw$, $ux$, $vw$, and $vx$ have each at least $5$ available $2$-colors, while the edge $uv$ has at least $7$.
	This means that we are always able to extend $\pi$ to $G$.
	Thus, $C$ is not adjacent to any $3$-cycle.
	
	Hence, the third neighbors of $u$, $v$, and $w$ (denote them $u'$, $v'$, and $w'$, resp.), are all distinct.
	Let $G'$ be the graph obtained from $G$ by removing $C$ and adding a new vertex $x$ adjacent to the vertices $u'$, $v'$, and $w'$.
	Let $\pi$ be a coloring of $G'$ and let $a = \pi(u'x)$, $b = \pi(v'x)$, and $c = \pi(w'x)$.
	Let $\varphi$ be a partial coloring of $G$ induced by $\pi$, and set $\varphi(u'u) = a$, $\varphi(v'v) = b$, and $\varphi(w'w) = c$.
	Notice that only the edges $uv$, $uw$, and $vw$ are not colored yet in $\varphi$, 
	and each of them has at most $7$ colored neighbors in its $2$-edge-neighborhood.
	If $0 \in \set{a,b,c}$, say $a = 0$, then for each edge of $C$ there are two available $2$-colors, 
	and moreover, the edge $vw$ can be colored with $0$, so $\varphi$ can be extended to all the edges.
	
	Hence, we may assume $a$, $b$, and $c$ are all $2$-colors and moreover, 
	they all have the same one $2$-color available, otherwise $\varphi$ can be extended to all the edges using the color $0$
	and two of the available $2$-colors.
	This means that every non-colored edge must see the same $7$ colors in its colored $2$-edge-neighborhood.	
	\begin{figure}[ht]	
		$$
			\includegraphics{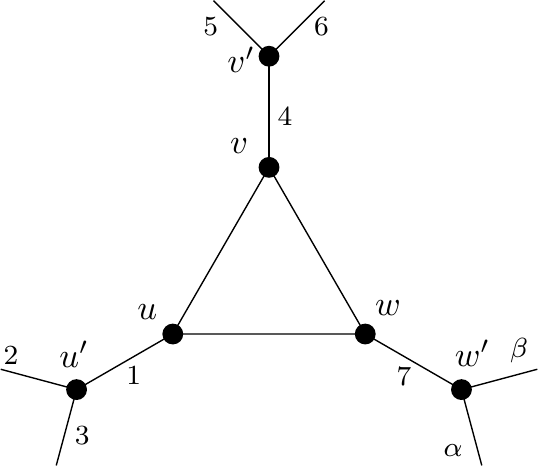}
		$$
		\caption{In a $3$-cycle where all three pendent edges are colored with a $2$-color, we cannot forbid the same $2$-color on all three edges.}
		\label{fig:128:3cycle}
	\end{figure}	
	Suppose $uv$ has $7$ distinct colors in its $2$-edge-neighborhood (as depicted in Fig.~\ref{fig:128:3cycle}).
	Then, in order to  have the same forbidden colors for $uw$ and $vw$, it must hold that $\set{\alpha,\beta} = \set{5,6}$
	and $\set{\alpha,\beta} = \set{2,3}$, respectively. We obtain a contradiction and so $\varphi$ can always be extended to all the edges of $G$.
\end{proofclaim}

\begin{claim}
	\label{128_girth5}
	$G$ does not contain $4$-cycles.
\end{claim}

\begin{proofclaim}
	We again proceed by contradiction. Suppose there is a $4$-cycle $C = uvwz$ in $G$.
	Let $u'$, $v'$, $w'$, and $z'$ be the neighbors of $u$, $v$, $w$, and $z$, respectively which do not belong to $C$. 
	Since $G$ has girth at least $4$, the eight edges $uv$, $vw$, $wz$, $zu$, $uu'$, $vv'$, $ww'$, and $zz'$ are distinct.
	Note that it is possible for the vertices $u'$ and $w'$ (resp. $v'$ and $z'$) to be equal; in such a case, 
	there is at least one more $2$-color available for the uncolored edges at distance two from this vertex. 
	This counter-balance the fact that we may need to use one more $2$-color on the edges incident with this vertex. 
	Therefore, we can assume that $u'$, $v'$, $w'$ and $z'$ are distinct. 
	%\todo{As written, $u'$ and $w'$ can be equals in all generality. I wrote some sentences to say that this does not realy matter. What do you think?}
	%the neighbors of the vertices of $C$, not contained in $C$,
	%(we denote them $u'$, $v'$, $w'$, and $z'$, respectively) are all distinct.
	We construct the graph $G' = (G \setminus V(C)) \cup \set{u'w',v'z'}$ (see the right graph in Fig.~\ref{fig:128:4cycle}).
	\begin{figure}[ht]	
		$$
			\includegraphics{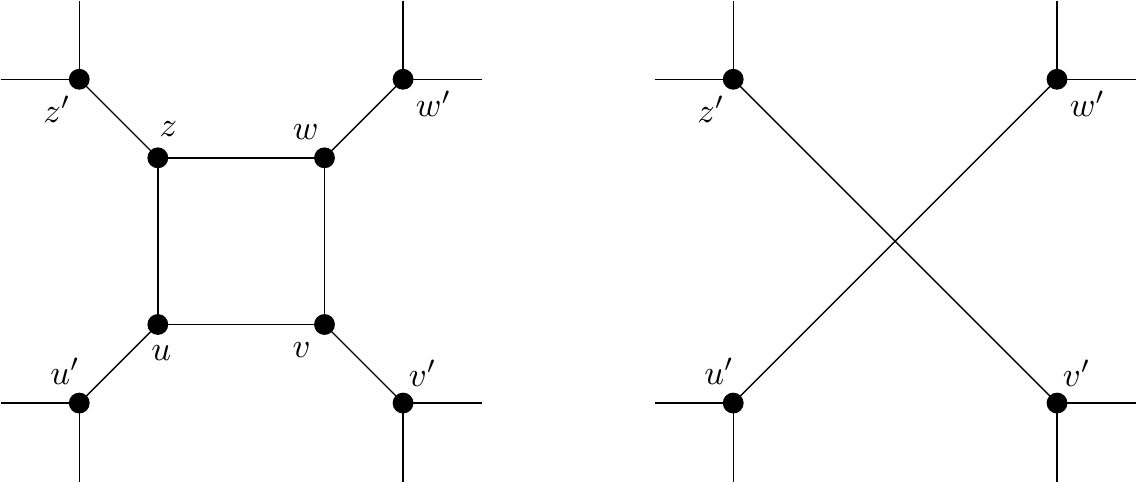}
		$$
		\caption{A $4$-cycle with its neighborhood in $G$ (left), and its replacement in the graph $G'$ (right).}
		\label{fig:128:4cycle}
	\end{figure}		
	By the minimality, $G'$ admits a good $(1,2^8)$-packing edge-coloring $\pi$, and we show that we can always extend $\pi$ to all the edges of $G$.
	We consider three cases regarding the colors of the edges $u'w'$ and $v'z'$ in $\pi$.	
	\begin{itemize}
		\item[$(i)$] \textit{Both, $u'w'$ and $v'z'$, are colored with the $1$-color $0$}. \quad
			Then the coloring $\varphi$ of $G$ induced by $\pi$ has only the edges of $C$ non-colored, 
			while the four pendent edges of $C$ (the edges with one end-vertex in $C$) are colored with $0$.
			This means that there are at least $4$ available $2$-colors for every edge of $C$,
			so we can complete the coloring by Theorem~\ref{th:Hall}.
		\item[$(ii)$] \textit{One of the edges, say $u'w'$, is colored with $0$, and the other with a $2$-color, say $1$.} \quad
			Then, in the induced coloring $\varphi$, we have $\varphi(u'u) = \varphi(w'w) = 0$ and $\varphi(v'v) = \varphi(z'z) = 1$.
			Now, the color $1$ appears on two edges of every $2$-edge-neighborhood of the edges of $C$. 
			So, each of them has at least $3$ available $2$-colors. We consider two subcases. 			
			If the union of all sets of available colors contains at least $4$ distinct colors, we can always choose 
			distinct colors for all the edge, by Hall's theorem. 
			
			So, we may assume that all four edges have the same set of $3$ available colors, say $\set{6,7,8}$.
			This means that on the edges incident to $u'$ and $v'$ there are colors $2$, $3$, $4$, and $5$.
			The same four colors must appear on the edges incident to $v'$ and $w'$, but this implies that
			at least two pairs of edges of the same $2$-color are at distance $2$ in $G'$, a contradiction.
		\item[$(iii)$] \textit{Both, $u'w'$ and $v'z'$, are colored with some $2$-color.} \quad
			In this case, we may color two opposite edges of $C$ with the color $0$.
			The remaining two non-colored edges have at least two available colors each, so we can always complete the coloring.
			This completes the proof of the claim.
	\end{itemize}
\end{proofclaim}

\begin{claim}
	\label{128_noCycles}
	$G$ contains no cycle of length at least $5$.
\end{claim}

\begin{proofclaim}
	Suppose the contrary and let $C = u_1u_2\dots u_n$ be a minimal induced $n$-cycle in $G$, with $n \ge 5$.
	For every $i$, $1\le i \le n$, let $u_i'$ be the neighbor of the vertex $u_i$ not in $C$, and let $G' = G \setminus V(C)$. 
	Note that the $u_i'$ are pairwise distinct by the minimality of $C$.
	Then, by the minimality of $G$, there is a good $(1,2^8)$-packing edge-coloring $\pi$ of $G'$.
	Since $\pi$ is good, no $u_i'$ is incident with the color $0$.
	So, in the coloring $\varphi$ of $G$ induced by $\pi$, we can color every edge $u_iu_i'$ with $0$.
	In this way, only the edges of $C$ are left non-colored and each edge of $C$ has at least four $2$-colors available. 
	
	Suppose first that $n = 5$.
	We can color $C$, except if all five edges have the same four $2$-colors available. 
	If we are in this case, 
	then suppose that $1$ and $2$ are the two colors on the edges incident to $u_1'$, 
	and $3$ and $4$ are the two colors on the edges incident to $u_2'$. 
	Then $\{1,2\}$ must also be on the edges incident to $u_3'$, 
	$\{3,4\}$ on the edges incident to $u_4'$, 
	and again $\{1,2\}$ on the edges incident to $u_5'$. 
	Thus the edge $u_1u_5$ has five available $2$-colors, a contradiction.
	
	If $n \geq 6$, then we can complete the coloring by Theorem~\ref{thm:listcyc}, a contradiction.
\end{proofclaim}

By Claims~\ref{128_deg2-}-\ref{128_noCycles},  $G$ is a cubic bridgeless graph without cycles, a contradiction.
This concludes the proof of Theorem~\ref{thm:128_good}.
\end{proof}

%
%
%	(1,2^7)-packing edge-coloring
%
%
\section{Proof of the case $(d)$ of Theorem~\ref{thm:main}}
\label{sec:1classI}

We split this section into two parts. 
First, we introduce notation and auxiliary results, and then use these to prove the case $(d)$ of Theorem~\ref{thm:main} in a stronger setting.

\subsection{Auxiliary results}

To show that certain graphs are strongly colorable from given lists, we will use Combinatorial Nullstellensatz, i.e., Theorem~\ref{thm:null}. 
For this purpose we introduce the following.
For two positive integers $k$ and $\ell$, where $k \le \ell$, we define the polynomial $P_{k,\ell}$ as follows:
\begin{equation}
	\label{eq:path}
	P_{k,\ell}(X_k, \dots, X_\ell)=(X_{k+1}-X_{k}) \cdot \prod_{i=k+2}^\ell (X_{i}-X_{i-2})(X_{i}-X_{i-1})\,.
\end{equation}
If $k=\ell$, by convention we take $P_{k,\ell}(X_k) = 1$.
Furthermore, for a monomial $m$, 
we denote by $p_{k,\ell}(m)$ the coefficient of $m$ in the polynomial $P_{k,\ell}$.

\begin{proposition} 
	\label{claim:pkl-nullstellensatz}
	For $k +2 \leq \ell$, we have the following equalities:
	\begin{align}
		p_{k,\ell}\paren{X_k\paren{\prod_{i=k+1}^{\ell -2}X_{i}^2} X_{\ell -1} X_\ell} &= 
		  \left\{
			  \begin{aligned}
				-1 & \mbox{ if } \ell -k \equiv 0 \bmod 3,\\
				1 & \mbox{ if } \ell -k \equiv 1 \bmod 3,\\
				0 & \mbox{ if } \ell -k \equiv 2 \bmod 3,\\
			  \end{aligned}
			\right.
			\label{eq:xpath1}
	\intertext{~}
		p_{k,\ell}\paren{X_k\paren{\prod_{i=k+1}^{\ell -1}X_{i}^2}} &= 
		  \left\{
			  \begin{aligned}
				0 & \mbox{ if } \ell -k \equiv 0 \bmod 3,\\
				-1 & \mbox{ if } \ell -k \equiv 1 \bmod 3,\\
				1 & \mbox{ if } \ell -k \equiv 2 \bmod 3,\\
			  \end{aligned}
			\right. 
			\label{eq:xpath2}
	\intertext{~}			
		p_{k,\ell}\paren{X_k\paren{\prod_{i=k+1}^{\ell -2}X_{i}^2} X_{\ell -1} X_\ell} 
			&= - ~p_{k,\ell}\paren{X_k X_{k+1}\paren{\prod_{i=k+2}^{\ell -1}X_{i}^2}  X_\ell}
			\label{eq:xpath3},
	\intertext{~}
		p_{k,\ell}\paren{X_k\paren{\prod_{i=k+1}^{\ell -1}X_{i}^2}} 
			&= - ~p_{k,\ell}\paren{\paren{\prod_{i=k+1}^{\ell -1}X_{i}^2} X_\ell}.
			\label{eq:xpath4}
	\end{align}
\end{proposition}

\begin{proof}
	First, note that by shifting the indices we can assume, without loss of generality, that $k=1$.
	Next, let 
	$a_\ell = p_{1,\ell}(X_1X_{2}^2\dots X_{\ell -2}^2 X_{\ell -1} X_\ell)$ and 
	$b_\ell = p_{1,\ell}(X_1X_{2}^2\dots X_{\ell -2}^2 X_{\ell -1}^2)$. 
	By expanding the factor $(X_{\ell} - X_{\ell-2})(X_{\ell} - X_{\ell-1})$ of $P_{1,\ell}$, 
	we obtain the following equalities on $a_\ell$ and $b_\ell$ for $\ell \geq 3$:
	\begin{align*}
		 a_\ell &= -p_{1,\ell-1}(X_1X_2^2\dots X_{\ell-3}^2X_{\ell -2}X_{\ell -1}) - p_{1,\ell-1}(X_1X_2^2\dots X_{\ell-3}^2X_{\ell -2}^2)\\
		 &= -a_{\ell-1} - b_{\ell -1},\\
		 b_\ell &= p_{1,\ell-1}(X_1X_2^2\dots X_{\ell-3}^2X_{\ell -2}X_{\ell -1})\\
		 &= a_{\ell -1}.
	\end{align*}
	Thus, 
	$$
		a_\ell = -a_{\ell - 1} - a_{\ell -2}\,.
	$$
	Moreover, $a_1 = 1$ and $a_2 = 0$, thus $a_3 = -1$. 
	By induction, we infer Equalities~\eqref{eq:xpath1} and~\eqref{eq:xpath2} for $a_{\ell}$ and $b_{\ell}$. 
	
	Symmetrically, by expanding the factor $(X_{3} - X_{1})(X_{2} - X_{1})$ of $P_{1,\ell}$,
	we infer analogous recurrences, and consequently Equalities~\eqref{eq:xpath3} and~\eqref{eq:xpath4} follow.
	We omit the proof.
\end{proof}

A graph with $k$ distinct edges $e_1,\dots,e_k$ is an \textit{$(a_1,a_2,\dots ,a_k)$-graph}
if its $i$-th edge $e_i$ is associated with a list of colors $L_i$ of size at least $a_i$. 
We say that an $(a_1,a_2,\dots,a_k)$-graph is \textit{strongly choosable} (or \textit{$(a_1,a_2,\dots,a_k)$-choosable}) 
if it admits a strong edge-coloring of its edges verifying that the color of $e_i$ belongs to $L_i$ for every assignment of $L_i$'s.
%For a path $P_n=u_0u_1\dots u_n$ (resp. cycle $C_n$), we denote its edges by $e_{i+1} = u_iu_{i+1}$ (and in the case of a cycle, the edge $u_0u_n$ is denoted by $e_{n+1}$). 
%We say that $P_n$ (resp. $C_n$) , where $ k= n$ (resp. $k=n+1$), 
%if every $(a_1,a_2,\dots,a_k)$-graph obtained by assigning lists to $P_k$ (resp. $C_k$) is strongly choosable.

%A path (resp. cycle) of length $k$ is an \textit{$(a_1,a_2,\dots, a_k)$-path} (resp. an \textit{$(a_1,a_2,\dots, a_k)$-cycle}) 
%if its $i$-th edge $e_i$ is associated with a list of colors $L_i$ of size at least $a_i$. 
%We say that an $(a_1,a_2,\dots, a_k)$-path (resp. $(a_1,a_2,\dots, a_k)$-cycle) is \textit{strongly choosable} 
%if it admits a strong edge-coloring verifying that the color of the $e_i$ belongs to $L_i$.
We often abbreviate the notation by joining consecutive terms of the same value, 
e.g., a $(2,3,3,3,1)$-path is abbreviated as a $(2,3^3,1)$-path.

Now, we show strong choosability of several configurations that will be used later in the proof.

\begin{lemma}
	\label{lem:path_choose}
	For any positive integer $\ell$, $\ell \ge 3$, a path of length $\ell$ is	
	\begin{itemize}
		\item[$(a)$] $(2,2,3^{\ell-3},2)$-choosable if $\ell \not \equiv 0 \bmod 3$;
		\item[$(b)$] $(2,3^{\ell-2},2)$-choosable.
	\end{itemize}
\end{lemma}

\begin{proof}
	Let $P$ be an $\ell$-path with the consecutive edges $e_1,\dots,e_\ell$, 
	and each edge $e_i$ has a list of available colors $L_i$ for every $i$, $1 \le i \le \ell$.
	Moreover, to each edge $e_i$, $1 \le i \le \ell$, we associate the variable $X_i$.

	First, consider the case $(a)$.
	By Theorem~\ref{thm:null}, we have that if the coefficient of $X_1 X_2 \paren{\prod_{i=3}^{\ell-1}X_{i}^2} X_\ell$ is non-zero, 
	then there is a solution $(x_1,\dots,x_\ell) \in L_1 \times \dots \times L_\ell$ such that $P_{1,\ell}(x_1,\dots,x_\ell) \neq 0$. 
	By Equation~\eqref{eq:xpath1} of Proposition~\ref{claim:pkl-nullstellensatz},
	this coefficient is non-zero if and only if $\ell - 1 \not\equiv 2 \bmod 3$, 
	thus only in the case when $\ell$ is not a multiple of $3$.
	This proves the case $(a)$.

	Now, we proceed with the case $(b)$.
	If $P$ is a $(2,3^{\ell-2}, 2)$-path, then it is also a $(2,2,3^{\ell-3}, 2)$-path.
	Thus, by the case $(a)$, it suffices to consider the case where $\ell$ is a multiple of $3$. 
	By Equation~\eqref{eq:xpath2} of Proposition~\ref{claim:pkl-nullstellensatz}, 
	the coefficient of $X_1\paren{\prod_{i=2}^{\ell - 1}X_{i}^2}$ is $1$ if $\ell \equiv 0 \bmod 3$, 
	and so, by Theorem~\ref{thm:null}, $P$ is colorable from its lists.
	This completes the proof.	
\end{proof}

Note that in the case $(b)$ of Lemma~\ref{lem:path_choose} we proved a stronger result; 
namely, a $(2,3^{\ell-2},1)$-path is strongly choosable if $\ell \equiv 0 \mod 3$.

\begin{proposition}
	\label{prop:trick}
	For a positive integer $\ell$,
	let $X$ be a set of colors with $|X| = 3$, 
	and $P$ be a $(2,3,\ldots,3,a)$-path of length $3\ell + 1$, 
	with $a \in \{2,3\}$, such that $L_i \subseteq X$ for every $i \in \{1,\ldots,3\ell + 1\}$. 
	If $\sigma$ is a strong edge-coloring of $P$ with $\sigma(e_i) \in L_i$, 
	then $\sigma(e_1) = \sigma(e_{3\ell+1}) \in L_1 \cap L_{3\ell+1}$.
\end{proposition}

\begin{proof}
	Without loss of generality, let $X=\{1,2,3\}$, $L_0=\{1,2\}$, and $\sigma(e_1)=1$. 
	Then, $\sigma(e_2),\sigma(e_3) \in \{2,3\}$, and hence $\sigma(e_4)=1$. 
	By induction, $\sigma(e_{3j+1}) = 1$ and $\sigma(e_{3j-1}), \sigma(e_{3j}) \in \{2,3\}$. 
	Thus, $\sigma(e_{3\ell+1}) = 1 = \sigma(e_1)$.
\end{proof}

%
%	El Diablo
%

Given an $n$-path $P=u_0u_1\dots u_n$, 
we define the graph $D_n$ as the graph obtained from $P$ by adding the edge $u_1u_{n-1}$ (see Figure~\ref{fig:eldiablo}).
In $D_n$, we will only color the edges of the initial path but with restrictions established also by the edge $u_1u_{n-1}$.
%\cmt{and we will say that $D_n$ is $(a_1,a_2,\dots,a_{n-1})$-choosable  
%if every $(a_1,a_2,\dots ,a_{n-1})$-graph obtained by assigning lists to $D_n$ is strongly choosable.}
\begin{figure}[htp!]	
	$$
		\includegraphics{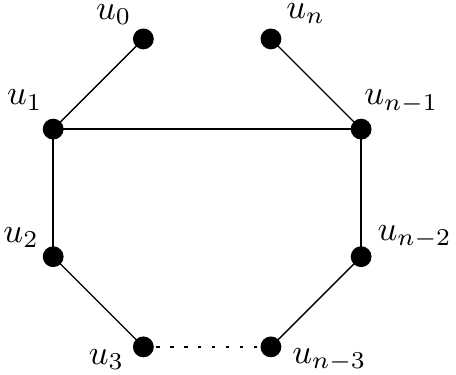}
	$$
	\caption{The graph $D_n$. Note that we do not color the edge $u_1u_{n-1}$ when considering this configuration.}
	\label{fig:eldiablo}
\end{figure}

Now, we define polynomials that will be used in proving strong choosability of $D_n$.
For an integer $n$, $n \ge 5$, we define the following:
$$
	C_n(X_1,\dots,X_{n})= (X_1-X_{n-1})(X_1-X_{n})(X_2-X_{n})(X_2-X_1) \mathlarger{\prod}_{i=3}^{n} \big( (X_{i}-X_{i-2})(X_{i}-X_{i-1}) \big),
$$
and
$$
	D_n(X_1,\dots,X_{n})=C_n \cdot (X_2-X_{n-1})\,.
$$
Note that they correspond to the coloring polynomial of the edges of $C_n$ and $D_n$.
For example, if $(x_1, \dots, x_{n})$ is a non zero solution of $C_n$, 
then it corresponds to a valid strong edge-coloring of the cycle of length $n$ where $c(e_i) = x_i$.
For a monomial $m$, by $c_n(m)$ we denote the coefficient of $m$ in $C_n$. 
Similarly, $d_n(m)$ is the coefficient of $m$ in $D_n$.

\begin{lemma}
	\label{lem:el_diablo}
	For every integer $n$, $n \geq 10$, $D_n$ is $(3,4,2,3,5,3,2,3^{n-9},4,2)$-choosable.
\end{lemma}

%\cmt{BORUT - naming edge $u_i u_{i+1} = e_i$ and not $e_{i+1}$ makes sense with cycles. Redo...}

\begin{proof}
By Theorem~\ref{thm:null}, it suffices to show that 
$$
	d_n(X_1^{2}X_2^{3}X_3X_4^{2}X_5^{4}X_6^{2}X_7X_8^{2}\ldots X_{n-2}^{2}X_{n-1}^{3}X_{n}) \neq 0\,.
$$
By the definition of $D_n$, we have that
\begin{align*}
	& d_n(X_1^{2} X_2^{3} X_3 X_4^{2} X_5^{4} X_6^{2} X_7 X_8^{2} \ldots X_{n-2}^{2} X_{n-1}^{3} X_{n})  \\ 
	=& \ c_n(X_1^{2} X_2^{2} X_3 X_4^{2} X_5^{4} X_6^{2} X_7 X_8^{2} \ldots X_{n-2}^{2} X_{n-1}^{3} X_{n}) 
	   - c_n(X_1^{2} X_2^{3} X_3 X_4^{2} X_5^{4} X_6^{2} X_7 X_8^{2} \ldots X_{n-2}^{2} X_{n-1}^{2} X_{n})
\end{align*}
Let 
$$
	\alpha = c_n(X_1^{2} X_2^{2} X_3 X_4^{2} X_5^{4} X_6^{2} X_7 X_8^{2} \ldots X_{n-2}^{2} X_{n-1}^{3} X_{n})
$$ 
and 
$$
	\beta = \um c_n(X_1^{2} X_2^{3} X_3 X_4^{2} X_5^{4} X_6^{2} X_7 X_8^{2} \ldots X_{n-2}^{2} X_{n-1}^{2}X_{n}).
$$

We consider the values of $\alpha$ and $\beta$ separately.
\begin{claim}
	\label{cl:alpha}
	For every integer $n \ge 10$, we have that $\alpha = 0$.
\end{claim}

\begin{proofclaim}
Suppose to the contrary that there is $n$ for which $\alpha \neq 0$.
Then by Theorem~\ref{thm:null}, a $(3,3,2,3,5,3,2,3^{n-9},4,2)$-cycle is strongly choosable. 
To reach a contradiction, we will find a $(3,3,2,3,5,3,2,3^{n-9},4,2)$-cycle for every $n$, which is not strongly choosable. 

%\cmt{Recall that the edges of our cycle are $e_1, \dots, e_{n}$.
%According to the value of $n$, we put different list on our cycle.}

\begin{itemize}
	\item[$(a)$] Let $n \equiv 0 \bmod 3$. Set 
		$L_1=\{1,2,3\}$,
		$L_2=\{1,2,3\}$,
		$L_3=\{2,3\}$,
		$L_4=\{1,2,3\}$,
		$L_5=\{1,2,3,4,5\}$,
		$L_6=\{1,3,4\}$,
		$L_7=\{1,4\}$,
		$L_8=\{1,3,4\},\ldots,L_{n-3}=\{1,3,4\}$,
		$L_{n-2}=\{2,3,4\}$,
		$L_{n-1}=\{1,2,3,4\}$,
		$L_{n}=\{1,2\}$.

		By Proposition~\ref{prop:trick}, using the path $e_3,e_2,e_1,e_{n}$, we have $\sigma(e_{n})= \sigma(e_3) = 2$. 
		Then, the color $2$ is forbidden in $L_{n-2}$, so we may assume that $L_{n-2}=\{3,4\}$. 
		By using Proposition~\ref{prop:trick} on the path $e_7,e_8,\dots,e_{n-3},e_{n-2}$, we infer $\sigma(e_{n-2})= \sigma(e_7) = 4$. 
		Now, since there are only colors $1$ and $3$ to color $e_1$, $e_2$, and $e_4$, by symmetry, 
		we may set $\sigma(e_1) = \sigma(e_4) = 1$ and $\sigma(e_2) = 3$.
		This forces $\sigma(e_{n-1}) = 3$. Now the colors for the edges $e_{n-3},\dots,e_8$, and $e_6$ are forced and $\sigma(e_6) = 1$. But, $e_4$ and $e_6$ are at distance $2$, so the cycle cannot be colored.

	\item[$(b)$] Let $n \equiv 1 \bmod 3$. Set 
		$L_1=\{1,2,3\}$,
		$L_2=\{1,2,3\}$,
		$L_3=\{2,3\}$,
		$L_4=\{2,3,5\}$,
		$L_5=\{1,2,3,4,5\}$,
		$L_6=\{1,4,5\}$,
		$L_7=\{1,4\}$,
		$L_8=\{1,3,4\},\ldots,L_{n-3}=\{1,3,4\}$,
		$L_{n-2}=\{1,2,4\}$,
		$L_{n-1}=\{1,2,3,4\}$, and
		$L_{n}=\{1,2\}$.

		Again, using the path $e_3,e_2,e_1,e_{n}$, by Proposition~\ref{prop:trick}, we have that $\sigma(e_{n})= c(e_3) = 2$.
		Similarly, using the path $e_7,e_8,\dots,e_{n-3}$, we infer $\sigma(e_{n-3}) \in \{1,4\}$. 
		Now, $\sigma(e_{n-2})$ and $\sigma(e_{n-3})$ both belong to $\{1,4\}$ and thus $\sigma(e_{n-1}) = 3$. 
		Hence, we deduce the following colors: $\sigma(e_1) = 1$, $\sigma(e_2) = 3$, $\sigma(e_4) = 5$, $\sigma(e_{n-4}) = 3,\dots,\sigma(e_{9}) = 3$,
		$\sigma(e_{7}),\sigma(e_{8}) \in \{1,4\}$, and $\sigma(e_{6}) = 5$. 
		But, $e_4$ and $e_6$ are at distance $2$, and so the cycle cannot be colored.

	\item[$(c)$] Let $n \equiv 2 \bmod 3$. Set 
		$L_1=\{1,2,3\}$,
		$L_2=\{1,2,3\}$,
		$L_3=\{2,3\}$,
		$L_4=\{1,2,3\}$,
		$L_5=\{1,2,3,4,5\}$,
		$L_6=\{1,3,4\}$,
		$L_7=\{1,3\}$,
		$L_{8}=\{1,3,4\},\ldots,L_{n-4}=\{1,3,4\}$,
		$L_{n-2}=\{1,3,4\}$,
		$L_{n-2}=\{1,2,3\}$,
		$L_{n-1}=\{1,2,3,4\}$, and
		$L_{n}=\{1,2\}$.

		Using the path $e_3,e_2,e_1,e_{n}$, by Proposition~\ref{prop:trick}, we again have that $\sigma(e_{n})= c(e_3) = 2$, and hence $\sigma(e_{n-2}) \in \set{1,3}$.
		Then, using the path $e_7,e_8,\dots,e_{n-4}$, we infer that $\sigma(e_{n-4}) \in \{1,3\}$. 
		Now, $\sigma(e_{n-2})$ and $\sigma(e_{n-4})$ both belong to $\{1,3\}$ thus $\sigma(e_{n-3}) = 4$. 
		After coloring these three edges, we can assume that $\sigma(e_{n-1}) = 1$. 
		It follows that $\sigma(e_1) = 3$, $\sigma(e_2) = 1$, $\sigma(e_4) = 3$, $\sigma(e_{n-2}) = 3$, 
		$\sigma(e_{n-4}) = 1,\dots,\sigma(e_{9}) = 3$, $\sigma(e_{8}) =4$, $\sigma(e_{7}) = 1$, and $\sigma(e_{6}) = 3$. 
		But $e_4$ and $e_6$ are at distance $2$, and so the cycle cannot be colored.
\end{itemize}
Since there are non-colorable cycles for every length $n$, we obtain a contradiction. Thus $\alpha = 0$.
\end{proofclaim}

\begin{claim}
	\label{cl:beta}
	For every integer $n \ge 12$, we have that $\beta \neq 0$.
\end{claim}

\begin{proofclaim}
	We first define the polynomial $Q_n$ as
	\begin{align*}
		Q_n = &(X_{n-1} - X_{n-2})(X_{n-1} - X_{n-3})(X_{n} - X_{n-1})(X_{n} - X_{n-2})\\
		&(X_{1} - X_{n-1})(X_{1} - X_{n})(X_{2} - X_{n})
	\end{align*}
	Observe that $C_n = Q_n \times P_{1,n-2}$.	
	Since $P_{1,n-2}$ does not contain $X_{n-1}$ and $X_{n}$, 
	it suffices to find the coefficients of $Q_n$ having $X_{n-1}^2 X_{n}$ as a factor to calculate $\beta$.
	Let us write $Q_n = \sum\limits_{i,j} X_{n-1}^i X_{n}^j R_{n,i,j}$, where $R_{n,i,j}$ is a polynomial in $X_1, \dots, X_{n-2}$. 
	We have
	\begin{align*}
		R_{n,2,1} =&  X_1^2 X_2 X_{n-3} + X_1^2 X_{n-3}X_{n-2} + X_1 X_2 X_{n-3} X_{n-2} \\
		&+ X_1^2 X_{n-2}^2  + X_1 X_{n-3} X_{n-2}^2 + X_2 X_{n-3} X_{n-2}^2\,.
	\end{align*}
	For a monomial $m$, we denote by $r_{n,2,1}(m)$ the coefficient of $m$ in $R_{n,2,1}(m)$.
	So, we may write $\beta$ as
	\begin{align}
		\begin{split}		
		\beta &= 
			\um c_n(X_1^{2} X_2^{3} X_3 X_4^{2} X_5^{4} X_6^{2} X_7 X_8^{2} \ldots X_{n-2}^{2} X_{n-1}^{2}X_{n}) \label{eq:beta} \\
			  &= - r_{n,2,1}(X_1^2 X_2 X_{n-3}) \cdot \overbrace{p_{1,n-2}(X_2^{2} X_3 X_4^{2} X_5^{4} X_6^{2} X_7 X_8^{2} \ldots X_{n-4}^{2} X_{n-3} X_{n-2}^{2})}^{\beta_1} \\
			  &-  r_{n,2,1}(X_1^2 X_{n-3} X_{n-2}) \cdot \overbrace{p_{1,n-2}(X_2^{3} X_3 X_4^{2} X_5^{4} X_6^{2} X_7 X_8^{2} \ldots X_{n-4}^{2} X_{n-3} X_{n-2})}^{\beta_2} \\
			  &-  r_{n,2,1}(X_1 X_2 X_{n-3} X_{n-2}) \cdot \overbrace{p_{1,n-2}(X_1 X_2^{2} X_3 X_4^{2} X_5^{4} X_6^{2} X_7 X_8^{2} \ldots X_{n-4}^{2} X_{n-3} X_{n-2})}^{\beta_3} \\
			  &-  r_{n,2,1}(X_1^2 X_{n-2}^2) \cdot \overbrace{p_{1,n-2}(X_2^{3} X_3 X_4^{2} X_5^{4} X_6^{2} X_7 X_8^{2} \ldots X_{n-4}^{2} X_{n-3}^2)}^{\beta_4} \\
			  &-  r_{n,2,1}(X_1 X_{n-3} X_{n-2}^2) \cdot \overbrace{p_{1,n-2}(X_1 X_2^{3} X_3 X_4^{2} X_5^{4} X_6^{2} X_7 X_8^{2} \ldots X_{n-4}^{2} X_{n-3})}^{\beta_5} \\
			  &-  r_{n,2,1}(X_2 X_{n-3} X_{n-2}^2) \cdot \overbrace{p_{1,n-2}(X_1^2 X_2^{2} X_3 X_4^{2} X_5^{4} X_6^{2} X_7 X_8^{2} \ldots X_{n-4}^{2} X_{n-3})}^{\beta_6}
		\end{split}
	\end{align}

	It remains to determine the coefficients of $P_{1,n-2}$ for the monomials appearing in~\eqref{eq:beta}.
	We compute them by reducing them in simpler forms.
	In particular, we make use of the facts that 
	$$
		P_{k,\ell}(X_k,\dots,X_\ell) = (X_{k+1}-X_k)(X_{k+2} - X_k) \cdot P_{k+1,\ell}(X_{k+1},\dots,X_\ell),
	$$
	% and
	$$
		P_{k,\ell}(X_k,\dots,X_\ell) = (X_{\ell}-X_{\ell-1})(X_{\ell} - X_{\ell-2}) \cdot P_{k,\ell-1}(X_{k},\dots,X_{\ell-1})\,.
	$$	
	and that $X_5$ must appear in four terms of $P_{1,n-2}$, meaning that  $X_5$ must be chosen in each of these terms when we expand the polynomial. 
	The same is true for $X_1$ when it is raised to the power $3$ as it appears in three terms.

	Let us now compute the six coefficients.
	\begin{align*}
		\beta_1 & = p_{1,n-2}(X_2^{2} X_3 X_4^{2} X_5^{4} X_6^{2} X_7 X_8^{2} \ldots X_{n-4}^{2} X_{n-3} X_{n-2}^{2}) \\
		& = p_{2,n-2}(X_2 X_4^{2} X_5^{4} X_6^{2} X_7 X_8^{2} \ldots X_{n-4}^{2} X_{n-3} X_{n-2}^{2}) \\
		& = \um p_{3,n-2}(X_4 X_5^{4} X_6^{2} X_7 X_8^{2} \ldots X_{n-4}^{2} X_{n-3} X_{n-2}^{2}) \\
		& = \um p_{4,n-2}(X_5^{3} X_6^{2} X_7 X_8^{2} \ldots X_{n-4}^{2} X_{n-3} X_{n-2}^{2}) \\
		& = \um p_{5,n-2}(X_5^{2} X_6 X_7 X_8^{2} \ldots X_{n-4}^{2} X_{n-3} X_{n-2}^{2}) \\
		&= \um p_{6,n-2}(X_6 X_7 X_8^{2} \ldots X_{n-4}^{2} X_{n-3} X_{n-2}^{2}) \\
		& = \um p_{6,n-3}(X_6 X_7 X_8^{2} \ldots X_{n-4}^{2} X_{n-3}) \\
		& = 
		\left \{ 
			\begin{array}{cl}
				\um 1 	&\mbox{if}\; n \equiv 0 \bmod 3; \\        
				1 		&\mbox{if}\; n \equiv 1 \bmod 3; \hspace{1cm} \text{by Proposition~\ref{claim:pkl-nullstellensatz}\eqref{eq:xpath3} and \eqref{eq:xpath1}}\\
				0 		&\mbox{if}\; n \equiv 2 \bmod 3.
		  	\end{array}
		\right.
	\end{align*}

	\begin{align*}
		\beta_2 &= p_{1,n-2}(X_2^{3} X_3 X_4^{2} X_5^{4} X_6^{2} X_7 X_8^{2} \ldots X_{n-4}^{2} X_{n-3} X_{n-2}) \\
		&= p_{2,n-2}(X_2^{2} X_4^{2} X_5^{4} X_6^{2} X_7 X_8^{2} \ldots X_{n-4}^{2} X_{n-3} X_{n-2}) \\
		&= p_{3,n-2}(X_4^{2} X_5^{4} X_6^{2} X_7 X_8^{2} \ldots X_{n-4}^{2} X_{n-3} X_{n-2}) \\
		&= p_{4,n-2}(X_4 X_5^{3} X_6^{2} X_7 X_8^{2} \ldots X_{n-4}^{2} X_{n-3} X_{n-2}) \\
		&= \um p_{5,n-2}(X_5^{2} X_6^{2} X_7 X_8^{2} \ldots X_{n-4}^{2} X_{n-3} X_{n-2}) \\
		&= \um p_{6,n-2}(X_6^{2} X_7 X_8^{2} \ldots X_{n-4}^{2} X_{n-3} X_{n-2}) \\
		&= \um p_{7,n-2}(X_7 X_8^{2} \ldots X_{n-4}^{2} X_{n-3} X_{n-2}) \\
		&=
		\left \{ 
			\begin{array}{cl}
				1 		&\mbox{if}\; n \equiv 0 \bmod 3; \\        
				\um 1 	&\mbox{if}\; n \equiv 1 \bmod 3; \hspace{1cm} \text{by Proposition~\ref{claim:pkl-nullstellensatz}\eqref{eq:xpath1}}\\
				0 		&\mbox{if}\; n \equiv 2 \bmod 3.
		  	\end{array}
		\right.		
	\end{align*}

	\begin{align*}
		\beta_3 &= p_{1,n-2}(X_1 X_2^{2} X_3 X_4^{2} X_5^{4} X_6^{2} X_7 X_8^{2} \ldots X_{n-4}^{2} X_{n-3} X_{n-2}) \\
		&= \um p_{2,n-2}(X_2^{2} X_4^{2} X_5^{4} X_6^{2} X_7 X_8^{2} \ldots X_{n-4}^{2} X_{n-3} X_{n-2}) \\
		&- p_{2,n-2}(X_2 X_3 X_4^{2} X_5^{4} X_6^{2} X_7 X_8^{2} \ldots X_{n-4}^{2} X_{n-3} X_{n-2}) \\
		&= \um \beta_2
			 + p_{3,n-2}(X_3 X_4 X_5^{4} X_6^{2} X_7 X_8^{2} \ldots X_{n-4}^{2} X_{n-3} X_{n-2}) \\
			&+ p_{3,n-2}(X_4^{2} X_5^{4} X_6^{2} X_7 X_8^{2} \ldots X_{n-4}^{2} X_{n-3} X_{n-2}) \\
		&= \um \beta_2
			- p_{4,n-2}(X_4 X_5^{3} X_6^{2} X_7 X_8^{2} \ldots X_{n-4}^{2} X_{n-3} X_{n-2})
			+ \beta_2 \\			
		&= p_{5,n-2}(X_5^{2} X_6^{2} X_7 X_8^{2} \ldots X_{n-4}^{2} X_{n-3} X_{n-2}) \\
		&= p_{6,n-2}(X_6^{2} X_7 X_8^{2} \ldots X_{n-4}^{2} X_{n-3} X_{n-2}) \\
		&= p_{7,n-2}(X_7 X_8^{2} \ldots X_{n-4}^{2} X_{n-3} X_{n-2}) \\
		&=
		\left \{ 
			\begin{array}{cl}
				\um 1 	&\mbox{if}\; n \equiv 0 \bmod 3; \\        
				1 		&\mbox{if}\; n \equiv 1 \bmod 3; \hspace{1cm} \text{by Proposition~\ref{claim:pkl-nullstellensatz}\eqref{eq:xpath1}}\\
				0 		&\mbox{if}\; n \equiv 2 \bmod 3.
		  	\end{array}
		\right.		
	\end{align*}
	Note that when reducing $p_{4,n-2}$, we used the fact that $X_5^4$ cannot appear in any monomial of $P_{4,n-2}$.
	Similarly, $X_k$ has at most power $2$ in monomials of $P_{k,\ell}$.

	\begin{align*}
		\beta_4 &= p_{1,n-2}(X_2^{3} X_3 X_4^{2} X_5^{4} X_6^{2} X_7 X_8^{2} \ldots X_{n-4}^{2} X_{n-3}^2) \\
		&= p_{2,n-2}(X_2^{2} X_4^{2} X_5^{4} X_6^{2} X_7 X_8^{2} \ldots X_{n-4}^{2} X_{n-3}^2) \\
		&= p_{3,n-2}(X_4^{2} X_5^{4} X_6^{2} X_7 X_8^{2} \ldots X_{n-4}^{2} X_{n-3}^2) \\
		&= p_{4,n-2}(X_4 X_5^{3} X_6^{2} X_7 X_8^{2} \ldots X_{n-4}^{2} X_{n-3}^2) \\
		&= \um p_{5,n-2}(X_5^{2} X_6^{2} X_7 X_8^{2} \ldots X_{n-4}^{2} X_{n-3}^2) \\
		&= \um p_{6,n-2}(X_6^{2} X_7 X_8^{2} \ldots X_{n-4}^{2} X_{n-3}^2) \\
		&= \um p_{7,n-2}(X_7 X_8^{2} \ldots X_{n-4}^{2} X_{n-3}^2) \\
		&= \um p_{7,n-3}(X_7 X_8^{2} \ldots X_{n-5}^2 X_{n-4} X_{n-3}) \\
		&=
		\left \{ 
			\begin{array}{cl}
				0 		&\mbox{if}\; n \equiv 0 \bmod 3; \\        
				1 		&\mbox{if}\; n \equiv 1 \bmod 3; \hspace{1cm} \text{by Proposition~\ref{claim:pkl-nullstellensatz}\eqref{eq:xpath1}}\\
				\um 1	&\mbox{if}\; n \equiv 2 \bmod 3.
		  	\end{array}
		\right.				
	\end{align*}
	%\cmt{to check if $n$ must be at least $12$}
	%$\beta_4=p_{6,8}(X_6)=-1$ if $n=11$.

	\begin{align*}
		\beta_5 &= p_{1,n-2}(X_1 X_2^{3} X_3 X_4^{2} X_5^{4} X_6^{2} X_7 X_8^{2} \ldots X_{n-4}^{2} X_{n-3}) \\
		&= \um p_{2,n-2}(X_2^{2} X_3 X_4^{2} X_5^{4} X_6^{2} X_7 X_8^{2} \ldots X_{n-4}^{2} X_{n-3}) \\
		&= \um p_{3,n-2}(X_3 X_4^{2} X_5^{4} X_6^{2} X_7 X_8^{2} \ldots X_{n-4}^{2} X_{n-3}) \\
		&= \um p_{4,n-2}(X_4^{2} X_5^{3} X_6^{2} X_7 X_8^{2} \ldots X_{n-4}^{2} X_{n-3}) \\
		&= \um p_{5,n-2}(X_5^{3} X_6^{2} X_7 X_8^{2} \ldots X_{n-4}^{2} X_{n-3}) \\
		&=\; 0 \hspace{1cm} \text{as there is no monomial with $X_5^3$ in $P_{5,n-2}$.}\\
	\end{align*}

	\begin{align*}
		\beta_6 &= p_{1,n-2}(X_1^2 X_2^{2} X_3 X_4^{2} X_5^{4} X_6^{2} X_7 X_8^{2} \ldots X_{n-4}^{2} X_{n-3}) \\
		&= p_{2,n-2}(X_2^{2} X_3 X_4^{2} X_5^{4} X_6^{2} X_7 X_8^{2} \ldots X_{n-4}^{2} X_{n-3}) \\
		&= p_{3,n-2}(X_3 X_4^{2} X_5^{4} X_6^{2} X_7 X_8^{2} \ldots X_{n-4}^{2} X_{n-3}) \\
		&= \um \beta_5 \\
		&= 0 \\
	\end{align*}

	Hence, inserting the values in~\eqref{eq:beta}, we obtain
	$$
		\beta = 
			\left \{ 
				\begin{array}{cl}
					1		&\mbox{if}\; n \equiv 0 \bmod 3; \\        
					\um 2 	&\mbox{if}\; n \equiv 1 \bmod 3; \\
					1		&\mbox{if}\; n \equiv 2 \bmod 3.
				\end{array}
			\right.					
	$$
\end{proofclaim}

From Claims~\ref{cl:alpha} and~\ref{cl:beta} it follows that
$$
	d_n(X_1^{2} X_2^{3} X_3 X_4^{2} X_5^{4} X_6^{2} X_7 X_8^{2} \ldots X_{n-2}^{2} X_{n-1}^{3} X_{n}) = \alpha - \beta \neq 0\,.
$$
This establishes the lemma.
\end{proof}

For smaller values of $n$, we prove another result.

\begin{lemma}
	\label{lem:el_diablo-small}
	For every integer $n$, $n \in \set{6,7,9,10}$, $D_n$ is $(3,4,3^{n-4},4,3)$-choosable.
	For $n \in \set{8,11}$, $D_n$ is $(3,3,3,5,3^{n-7},2,4,3)$-choosable.
\end{lemma}

\begin{proof}
	To prove the lemma, we again use Theorem~\ref{thm:null} and for each $n \in \set{6,7,8,9}$
	provide a monomial $m_n$ such that $d_n(m_n) \neq 0$ in $D_n$.
	\begin{itemize}
		\item{} If $n=6$, we have $d_6(X_1^2 X_2^3 X_3^2 X_4^2 X_5^2 X_6^2) = 2$;
		\item{} If $n=7$, we have $d_7(X_1^2 X_2^3 X_3^2 X_4^2 X_5^2 X_6^2 X_7^2) = 1$;
		\item{} If $n=8$, we have $d_8(X_1^2 X_2^2 X_3^2 X_4^3 X_5^2 X_6   X_7^3 X_8^2) = -1$;
		\item{} If $n=9$, we have $d_9(X_1^2 X_2^3 X_3^2 X_4^2 X_5^2 X_6^2 X_7^2 X_8^2 X_9^2) = 2$;
		\item{} If $n=10$, we have  $d_{10}(X_1^2 X_2^3 X_3^2 X_4^2 X_5^2 X_6^2 X_7^2 X_8^2 X_9^2X_{10}^2) = 1$;
		\item{} If $n=11$, we have $d_{11}(X_1^2 X_2^2 X_3^2 X_4^3 X_5^2 X_6^2 X_7^2 X_8^2 X_9   X_{10}^3 X_{11}^2) = -1$.
	\end{itemize}
	It is easy to verify that the degree of every variable is less than the number of available colors assumed by the lemma,
	and thus we infer the desired choosabilities of $D_n$.
\end{proof}

\subsection{Proof}

Recall that in the case $(d)$ of Theorem~\ref{thm:main}, we assume the graph is in class I.
In our proof, this is an important feature which enables us to confirm Conjecture~\ref{conj:main}$(b)$ for this class of graphs.
We again prove a stronger version of the theorem.
\begin{theorem}
	\label{th:127}
	Let $G$ be a graph of class I. 
	Then for every proper $3$-edge-coloring $\pi$ with colors $a$, $b$, and $c$, 
	and for every color $\alpha \in \set{a,b,c}$
	there exists a $(1,2^7)$-packing edge-coloring $\sigma$ such that the edges of color $\alpha$ in $\pi$ are colored with $0$ in $\sigma$.
\end{theorem}

For simplicity, we will refer to a $(1,2^7)$-packing edge-coloring $\sigma$ obtained from a proper $3$-edge-coloring $\pi$, 
in which the edges of color $\alpha$ in $\pi$ are colored with $0$ in $\sigma$, 
as an \textit{$\alpha$-induced coloring $\sigma_\pi^\alpha$}.

\begin{proof}
	We prove the theorem by contradiction.
	Let $G$ be a minimal counterexample to the theorem minimizing the sum $|V(G)| + |E(G)|$.
	Let $\pi$ be a proper $3$-edge-coloring (using colors $a$, $b$, and $c$) 
	and let the color $a$ be the color class for which there is no $(1,2^7)$-packing edge-coloring $\sigma$ (using colors in $\set{0,1,\dots,7}$ and $0$ being the $1$-color) 
	of $G$ such that all edges colored $a$ in $\pi$ are colored $0$ in $\sigma$. 
	
	We begin by establishing some structural properties of $G$.

\begin{claim}
	\label{cl:127_simple}
	$G$ is simple.
\end{claim}

\begin{proofclaim}
	Suppose there are vertices $u$ and $v$ in $G$ connected by at least two parallel edges.
	Remove one of the edges, call it $e$, between them (if possible, take the one colored with $a$) to obtain a smaller graph $G'$. 
	By the minimality of $G$, there is an $a$-induced coloring $\sigma_\pi^a$ of $G'$.
	If $e$ is colored by $a$ in $\pi$, then extend $\sigma_\pi^a$ by coloring $e$ with $0$. 
	Otherwise there is no edge colored $a$ between $u$ and $v$.
	Therefore, if $u'$ and $v'$ are the respective other neighbors of $u$ and $v$, when they exist, 
	then $uu'$ and $vv'$ are colored $a$. 
	Thus, $A_2(e) \geq 2$ and we can extend $\sigma_\pi^a$ to $G$, a contradiction.
	Loops can be treated similarly.
\end{proofclaim}

\begin{claim}
	\label{cl:127_min3}
	$G$ is cubic.
\end{claim}

\begin{proofclaim}
	Suppose the contrary and let $u$ be a $2^-$-vertex. 
	By the minimality, there is an $a$-induced coloring $\sigma_\pi^a$ of $G-u$.
	
	Suppose first that $u$ is a $1$-vertex with a neighbor $v$. 
	If $uv$ is colored with $a$ in $\pi$, then we color $uv$ with $0$ and hence extend $\sigma_\pi^a$ to all the edges of $G$.
	If $uv$ is not colored $a$ in $\pi$, then there are at least three available $2$-colors for $uv$ and $\sigma_\pi^a$ can be extended to $G$. 
	
	So, we may assume that $u$ is a $2$-vertex and let $v$ and $w$ be its two neighbors.
	We consider two subcases. 
	First, if in $\pi$ none of $uv$ and $uw$ are colored with $a$, then $A_2(uv) \geq 2$ and $A_2(uw) \geq 2$, and so we can color the two edges.
	Second, if in $\pi$ one of $uv$ and $uw$ is colored with $a$, say $uv$, 
	then we color $uv$ with $0$ and we obtain $A_2(uw) \geq 1$. Hence we can extend $\sigma_\pi$ to all the edges of $G$.	
\end{proofclaim}

Recall that $G$ being cubic implies that in $\pi$ every color appears at every vertex.

\begin{claim}
	\label{cl:127_triangle}
	$G$ does not contain $3$-cycles.
\end{claim}

\begin{proofclaim}
	Suppose the contrary and let $C = uvw$ be a $3$-cycle in $G$.
	Let $u'$, $v'$, and $w'$ be the neighbors of $u$, $v$, and $w$, respectively, not on $C$.
	Since $G$ is cubic, by Claim~\ref{cl:127_min3}, $u'$, $v'$, and $w'$ are $3$-vertices and in the coloring $\pi$
	exactly one of the edges $uu'$, $vv'$, and $ww'$ is colored with $a$, say $uu'$. 
	By the minimality, there is a coloring $\sigma_\pi^a$ of $G \setminus E(C)$ induced by $\pi$.
	We can extend it to the edges of $C$ in the following way.
	First, we color $vw$ with $0$. Next, observe there are at least two available $2$-colors for each of the edges $uv$ and $uw$, 
	hence we can always color them, a contradiction.
\end{proofclaim}

\begin{claim}
	\label{cl:127_square}
	$G$ does not contain $4$-cycles.
\end{claim}

\begin{proofclaim}
	Suppose the contrary and let $C = uvwz$ be a $4$-cycle in $G$. 
	Let $u'$, $v'$, $w'$, and $z'$ be the neighbors of $u$, $v$, $w$, and $z$, respectively, not on $C$ (see Fig.~\ref{fig:127:4cycle}).
	\begin{figure}[ht]
		$$
			\includegraphics{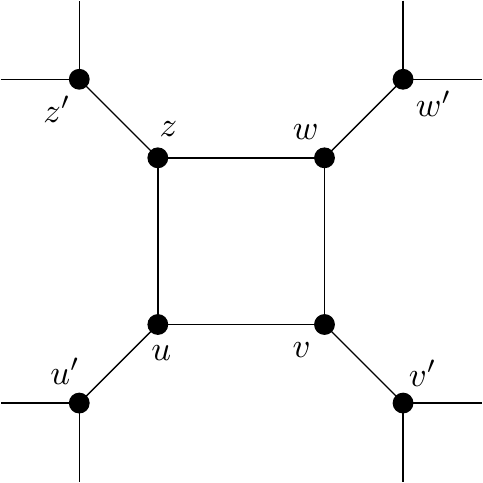}
		$$
		\caption{A $4$-cycle with its neighborhood in $G$.}
		\label{fig:127:4cycle}
	\end{figure}	
	By Claim~\ref{cl:127_min3} and Claim~\ref{cl:127_triangle}, 
	the vertices $u'$, $v'$, $w'$, and $z'$ are all of degree $3$ and the vertices $u'$ and $w'$ (resp. $v'$ and $z'$) are distinct. 
	Note that it is possible that $u' = w'$ (resp. $v'= z'$)
	but in such a case one would have even more $2$-colors available to color the cycle. %even if more constraints arise. 
	Hence, we may assume that $u'$, $v'$, $w'$, and $z'$ are distinct. %\todo{ Again we do not treat all cases but it is not important (?)}
	There are three non-symmetric possibilities for $\pi$ to assign colors to the edges 
	$uu'$, $vv'$, $ww'$, and $zz'$. We consider each of them separately.
	Before the case analysis, observe that it is not possible to have two opposite  pendent edges to $C$
	colored with $a$, and at least one edge of the other pair of the pendent edges not colored with $a$, 
	since the edges of $C$ could not be colored with three colors.
	\begin{itemize}
		\item[$(a)$] $uu'$, $vv'$, $ww'$, and $zz'$ are all colored with $a$ in $\pi$. \quad
			By the minimality, there is an $a$-induced coloring $\sigma_{\pi'}^a$ of $G'=(G \setminus V(C_4)) \cup \set{u'w',v'z'}$, 
			where $\pi'$ is a proper $3$-edge-coloring of $G'$ obtained from $\pi$ by coloring $u'w'$ and $v'z'$ by $a$.
			Now, consider the partial coloring of $G$ induced by $\sigma_{\pi'}^a$ 
			and color the edges $uu'$, $vv'$, $ww'$, and $zz'$ with $0$.
			In this way, every edge of $C$ has at most four $2$-colors in its $2$-edge-neighborhood 
			and thus at least three available $2$-colors. 
			If the union of the available colors of all three edges contains at least four colors, we can color the edges by the Hall's Theorem.
			
			So, we may assume that all four edges have the same three available colors, say $5$, $6$, and $7$.
			This implies that on the edges incident to $u'$ and $v'$ there are colors $1$, $2$, $3$, and $4$ (together with $0$ on the two edges incident with $C$).
			The same four colors must appear on the edges incident to $v'$ and $w'$. 
			But this means that at least two pairs of edges of the same $2$-color are at distance $2$ in $G'$, since $u'w' \in E(G')$, a contradiction.		
			
		\item[$(b)$] Two edges pendent to $C$ and one edge of $C$ are colored with $a$ in $\pi$, say $uu'$, $vv'$, and $wz$. \quad
			Consider the graph $G' = (G \setminus \set{u,v}) \cup \set{u'v'}$ and a proper $3$-edge-coloring $\pi'$ of $G'$ induced by $\pi$ by coloring $u'v'$ by $a$.
			By the minimality, there is an $a$-induced coloring $\sigma_{\pi'}^a$ of $G'$.
			
			Now, consider the coloring $\sigma$ of $G$ induced by $\sigma_{\pi'}$, 
			where only the edges $uz$, $vw$, and $uv$ remain non-colored.
			There are at most six $2$-colors in the $2$-edge-neighborhood of $uv$, so we may color it.
			After that, there are at most six $2$-colors in the $2$-edge-neighborhoods of $uz$ and $vw$, since $z'$ and $w'$ are each incident with one edge of color $0$.
			However, there are four distinct $2$-colors incident with the vertices $u'$ and $v'$ (recall, $u'v' \in E(G')$),
			and thus the union of available colors for $uz$ and $vw$ contains at least two colors, meaning that we can complete the coloring, a contradiction.
			
		\item[$(c)$] Two edges of $C$ are colored with $a$ in $\pi$, say $uv$ and $wz$. \quad
			Let $G' = G \setminus V(C)$ and let $\pi'$ be a proper $3$-edge-coloring of $G'$ induced by $\pi$.
			By the minimality, there is an $a$-induced coloring $\sigma_{\pi'}^a$ of $G'$.
			Now, consider the coloring $\sigma$ of $G$ induced by $\sigma_{\pi'}^a$ and color the edges $uv$ and $wz$ with $0$.
			For the non-colored edges, we have the following numbers of available colors:
			$A_2(uu') = A_2(vv') = A_2(ww') = A_2(zz') = 3$ and $A_2(uz) = A_2(vw) = 5$.
			To show that $\sigma$ can be extended to non-colored edges, we apply Theorem~\ref{thm:null} in the following way.
			First, associate the variables $X_1$, $X_2$, $X_3$, $X_4$, $X_5$, and $X_6$ 
			to the edges $uu'$, $uz$, $zz'$, $vv'$, $vw$, and $ww'$, respectively.
			The chromatic polynomial of a subgraph induced by the non-colored edges and setting adjacencies whenever two edges are at distance at most $2$
			is
			\begin{eqnarray*}
			f(X_1,\ldots,X_6) 	&=& 	(X_1-X_2)(X_1-X_3)(X_1-X_4)(X_1-X_5) \\
							  	&\times& (X_2-X_3)(X_2-X_4)(X_2-X_5)(X_2-X_6) \\
							  	&\times& (X_3-X_5)(X_3-X_6)(X_4-X_6)(X_4-X_5)(X_5-X_6)
			\end{eqnarray*}		
			Expanding the polynomial, we infer that the coefficient of the monomial 
			$X_1^2X_2^4X_3^2X_4^2X_5^3$ in $f$ equals $+1$. 
			Thus, by Theorem~\ref{thm:null}, we can extend $\sigma$ to all the edges of $G$, a contradiction.
	\end{itemize}
	This establishes the claim.
\end{proofclaim}

Hence, $G$ has girth at least $5$. We continue by considering properties of longer cycles in $G$.
First, we introduce some additional definitions.
A cycle colored only with the colors $b$ and $c$ is called a \textit{$bc$-cycle}.
Let $P = u_{0}u_{1}u_{2}u_{3}u_{4}u_{5}$ be a path of distinct vertices on a $bc$-cycle. 
Let $u_i'$ be the neighbor of $u_i$ such that $u_iu_i'$ is colored with $a$ in $\pi$ for $i \in \{0,\dots,5\}$. 
A \textit{$P$-crossing} is a pair $(G',\pi')$ such that $G' = G - E(P) - \{u_ju_j'\}_{1 \le j \le 4} + \{u_{1}'u_{3}', u_{2}'u_{4}'\}$  
and $\pi'$ is obtained from $\pi$ by coloring $u_{1}'u_{3}'$ and $u_{2}'u_{4}'$ with the color $a$ (see Figure~\ref{fig:crossing}).
\begin{figure}[htp!]
	$$
		\includegraphics{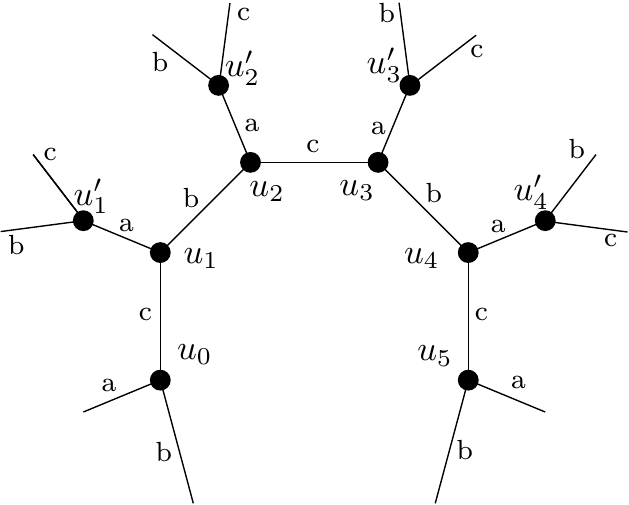} \quad\quad
		\includegraphics{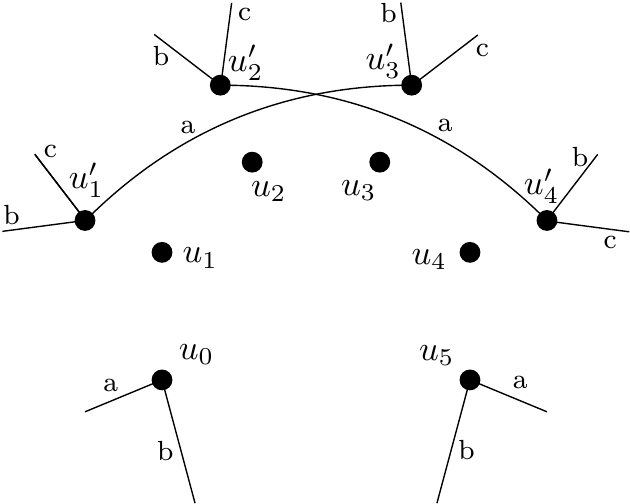}
	$$
	\caption{A path $P$ on which we perform a crossing in $G$ (left), 
		and the configuration in $G'$ (right).}
	\label{fig:crossing}
\end{figure}
By the minimality of $G$, there exists an $a$-induced coloring $\sigma_{\pi'}^a$ of $G'$.
The partial coloring $\sigma_{\pi}^a$ of $G$ induced by $\sigma_{\pi'}^a$ leaves uncolored the edges of $P$
and the four edges $u_ju_j'$ for $1 \le j \le 4$. 
Clearly, we can color the latter four edges with $0$, and so only the edges of $P$ need to be colored.
In the next claim, we give a useful property about their lists of available colors.

\begin{claim}
	\label{cl:127_crossing}
	Let $L_i$ be the set of $2$-colors available for the edge $u_iu_{i+1}$ of $P$, where $1 \le i \le 3$.
	In the coloring $\sigma_{\pi}^a$, one of the following properties hold:
	\begin{itemize}
		\item[$(1)$] $\abs{L_2} = 5$; or
		\item[$(2)$] there exists a color $x \in L_2$ such that $\abs{L_1 \setminus x} \ge 3$ and $\abs{L_3 \setminus x} \ge 3$; or
		\item[$(3)$] $\abs{L_2} \ge 4$ and $L_1 \cap L_3 = \emptyset$. 
			Moreover, there is a color $x \in L_1$ such that $x \notin L_2$, and a color $y \in L_3$ such that $y \notin L_2$.
	\end{itemize}
\end{claim}

\begin{proofclaim}
	Note first that the edges $u_1u_2$, $u_2u_3$, and $u_3u_4$ all 
	have four colored edges in their $2$-edge-neighborhoods.
	Without loss of generality, we may assume that $u_1'$ is incident with edges colored $1$ and $2$,
	$u_3'$ is incident with edges colored $3$ and $4$ (the colors are distinct as $u_1'$ and $u_3'$ are adjacent in $G'$).
	Denote by $S_1 = \set{x_1,x_2}$ the set of $2$-colors on the edges incident with $u_2'$,
	and by $S_2 = \set{x_3,x_4}$ the set of $2$-colors on the edges incident with $u_4'$ (see Figure~\ref{fig:claim7}). 
	Again, $S_1 \cap S_2 = \emptyset$.
	\begin{figure}[htp!]
		$$
			\includegraphics{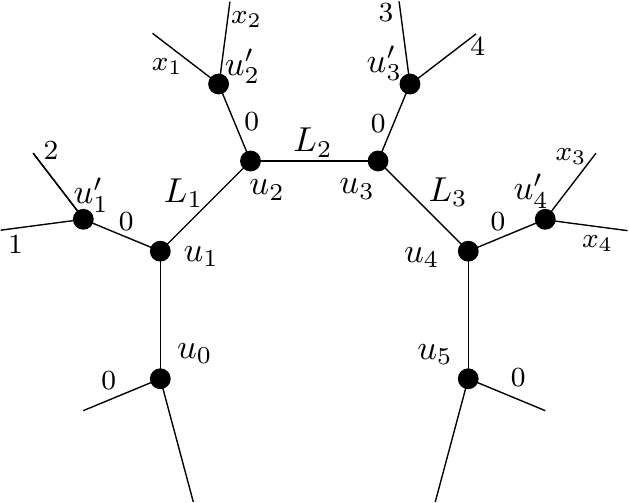}
		$$
		\caption{Configuration in $G$ for considering properties of $L_1$, $L_2$, and $L_3$.}
		\label{fig:claim7}
	\end{figure}
	
	We consider possibilities regarding the sets $S_1$ and $S_2$.
	Note that if $|L_1| \ge 4$ and $|L_3| \ge 4$, then we always can choose such a color in $L_2$ to satisfy condition $(2)$.
	Suppose first that $S_1 = \set{1,2}$. Then $|L_1| = 5$ and $L_2 = \set{5,6,7}$. 
	Since $S_2$ does not contain $1$ or $2$, $\set{1,2} \subset L_3$. 
	Thus, either $|L_3| \ge 4$ or $L_2$ contains a color which is not in $L_3$.
	
	Next, suppose that $S_1 = \set{1,3}$. Then $|L_1| \ge 4$. If $4 \in S_2$, then also $L_3 \ge 4$ and we are done.
	Otherwise, by symmetry, we may assume that $x_3 = 5$ and $x_4 \in \set{2,6}$. 
	Observe that $x_4 \in L_2$ and $x_4 \notin L_3$, thus setting $x = x_4$ gives us condition $(2)$ of the claim.
	
	Suppose now that $S_1 = \set{1,5}$. Again, $|L_1| \ge 4$. If $3 \in S_2$, then $L_3 \ge 4$ and we are done. 
	Thus, we may assume that $x_3 = 6$ and $x_4 \in \set{2,7}$.  
	As above, observe that $x_4 \in L_2$ and $x_4 \notin L_3$, thus setting $x = x_4$ gives us condition $(2)$.
	
	If $S_1 = \set{3,4}$, we are done as $|L_2| = 5$.
	
	If $S_1 = \set{3,5}$, we have $L_1 = \set{4,6,7}$ and $L_2 = \set{1,2,6,7}$.
	If $1 \in S_2$ or $2 \in S_2$, then we set $x=2$ or $x=1$, respectively, to obtain condition $(2)$.
	Thus, we may assume $x_3 = 6$ and $x_4 \in \set{4,7}$. If $x_4 = 4$, then we set $x=1$, 
	and if $x_4 = 7$, then $L_3 = \set{1,2,5}$, and we have condition $(3)$.
	
	Finally, if $S_1 = \set{5,6}$, then, by symmetry, $S_2 = \set{1,2}$, and setting $x=1$ gives us condition $(2)$. 
	This completes the proof.
\end{proofclaim}

\begin{claim}
	\label{cl:127_chords}
	There is no $bc$-cycle with chords in $G$.
\end{claim}

\begin{proofclaim}
	Suppose the contrary and let $C$ be a $bc$-cycle with a chord in $G$.
	Let $P = u_0\dots u_n$ be a path of $C$ such that $u_1u_{n-1}$ is a chord of $C$ and $P$ is the shortest path with this property in $C$.
	For every $i \in \set{0,\dots,n}$, denote by $u_i'$ the neighbor of $u_i$ such that $\pi(u_iu_i') = a$.  
	Note that $\pi(u_iu_i') = a$ implies that all $u_i'$s are pairwise distinct. Note that by defintion $u_1' = u_{n-1}$ and $u_{n-1}' = u_1$.
	We split the proof in two cases regarding the length of $P$.

	Suppose first that $n \geq 12$.
	Let $P' = u_2u_3u_4u_5u_6u_7$ and let $(G',\pi')$ be the $P'$-crossing and $\sigma_{\pi'}^a$ an $a$-induced coloring of $G'$.
	Let $\sigma_{\pi}^a$ be the partial coloring of $G$ induced by $\sigma_{\pi'}^a$.
	For every $i$, $i \in \set{3,4,5,6}$, color $u_iu_i'$ with $0$.
	Now, only the edges of $P'$ are non-colored. 
	To extend $\sigma_{\pi}^a$ to all edges of $G$, we first uncolor also the edges of $P$ that are already colored.	
	Next, for every $i$, $0 \le i \le n-1$, denote by $L_i$ the list of available $2$-colors of the edge $u_iu_{i+1}$ in $G$.
	Note that $\abs{L_0} \geq 3$, $\abs{L_1} \geq 4$, $\abs{L_{n-2}} \geq 4$, $\abs{L_{n-1}} \geq 3$, and $\abs{L_j} \geq 3$ for $2 \le j \le n-3$.
	By the minimality of $P$ the $2$-edge-neighborhood of an edge $e$ of $P$ contains the same non-colored edges 
	as in the graph $D_n$ obtained from $P$ by adding the edge $u_1u_{n-1}$. 
	Therefore, it suffices to color $D_n$ using the lists $L_i$ to extend $\sigma_{\pi}^a$.
	
	To show that we can color $D_n$, we make use of Claim~\ref{cl:127_crossing} applied to $P'$.
	In the case $(1)$, i.e., if $|L_4| = 5$, we can color $D_n$ by Lemma~\ref{lem:el_diablo}, a contradiction. 
	In the case $(2)$, i.e., if there exists a color $x \in L_4$ such that $\abs{L_3 \setminus x} \ge 3$ and $\abs{L_5 \setminus x} \ge 3$,
	we proceed as follows. We first color $u_4u_5$ with $x$. 
	We obtain, after updating the lists, the following: 
	$\abs{L_2} \geq 2$, $\abs{L_3} \geq 3$, $\abs{L_5} \geq 3$, $\abs{L_6} \geq 2$. 	
	Note that since $u_4u_5$ is already colored, we can simply assume $|L_4| \ge 5$ to be able to apply Lemma~\ref{lem:el_diablo} on the other edges of $D_n$.
	Hence, we again extend $\sigma_{\pi}^a$ to all edges of $G$, a contradiction.
	In the case $(3)$, i.e., if $\abs{L_4} \ge 4$ and $L_3 \cap L_5 = \emptyset$, 
	we use the color $x \in L_3$, which is not in $L_4$, to color $u_3u_4$.
	By doing this, we only decrease the number of available colors in $L_1$ and $L_2$ to $3$ and $2$, respectively.
	Next, we consecutively color $u_2u_3$, $u_1u_2$, and $u_0u_1$. 
	By doing this, we obtain $|L_4| \ge 3$, $|L_{n-2}| \ge 2$, and $|L_{n-1}| \ge 1$. 
	All the other non-colored edges still have at least $3$ available colors, 
	and hence we can extend $\sigma_{\pi}^a$ by coloring consecutively the edges $u_{n-1}u_n$, $u_{n-2}u_{n-1}$,\dots, $u_4u_5$, a contradiction.

	Hence, we may assume that $n < 12$. 
	By Claim~\ref{cl:127_square}, we also have that $n \geq 7$. 
	We first consider the case when $n \notin \set{8,11}$.
	Let $P' = u_1u_2u_3u_4u_5u_6$ and let $(G',\pi')$ be the $P'$-crossing  
	 and $\sigma_{\pi'}^a$ an $a$-induced coloring of $G'$.
	Let $\sigma_{\pi}^a$ be the partial coloring of $G$ induced by $\sigma_{\pi'}^a$.
	For every $i$, $i \in \set{2,3,4,5}$, color $u_iu_i'$ with $0$, and uncolor the colored edges of $P$.
	It is easy to see that all the edges have at least $3$ available colors and the edges $u_1u_2$ and $u_{n-2}u_{n-1}$ have at least $4$.
	Note that the edges of $P$ together with the edge $u_1u_{n-1}$ form the graph $D_n$, 
	which is $(3,4,3^{n-4},4,3)$-choosable by Lemma~\ref{lem:el_diablo-small} for $n \in \set{6,7,9,10}$. 
	Thus we can extend $\sigma_{\pi}^a$ to $G$, a contradiction. 

	So, we may assume that $n \in \set{8,11}$.
	Let $P' = u_1u_2u_3u_4u_5u_6$ and let $(G',\pi')$ be the $P'$-crossing and $\sigma_{\pi'}^a$ an $a$-induced coloring of $G'$.
	Let $\sigma_{\pi}^a$ be the partial coloring of $G$ induced by $\sigma_{\pi'}^a$.
	For every $i$, $i \in \set{2,3,4,5}$, color $u_iu_i'$ with $0$.
	Now, only the edges of $P'$ are non-colored. 	
	To extend $\sigma_{\pi}^a$ to all edges of $G$, we first uncolor the edges of $P$ that are already colored.
	As above, it suffices to find a list coloring of  the graph $D_n$ obtained from $P$ by adding the edge $u_1u_{n-1}$. 
	We will again apply Claim~\ref{cl:127_crossing}. 
	Suppose first that $|L_3| = 5$. Then, we can extend the coloring by Lemma~\ref{lem:el_diablo-small}, saying that $D_n$ is $(3,3,3,5,3^{n-7},2,3,4)$-choosable.
	Suppose now that there exists $x \in L_{3}$ such that $|L_{2} \setminus x| \ge 3$ and $|L_{4} \setminus x| \ge 3$. 
	We color $u_3u_4$ with $x$, we obtain $\abs{L_1} \geq 3$, $\abs{L_2} \geq 3$, $\abs{L_4} \geq 3$, $\abs{L_5} \geq 2$, and we can assume that $\abs{L_3} \geq 5$, 
	since it is already colored anyway. Hence, we can color $D_n$ by Lemma~\ref{lem:el_diablo-small}, a contradiction.
	Finally, suppose that $\abs{L_2} \ge 4$ and $L_1 \cap L_3 = \emptyset$.
	In this case, we color the edges of $P$ as follows.
	First, color $u_1u_2$ with a color that is not contained in $L_7$.
	Next, color $u_0u_1$ and then $u_2u_3$. 
	Note that at this point, $|L_3| \ge 2$, $|L_4| \ge 3$, $|L_5| \ge 3$, $|L_6| \ge 2$, and $|L_7| \ge 2$.
	This means that we can complete the coloring by using Lemma~\ref{lem:path_choose}$(a)$, a contradiction.	
\end{proofclaim}

\begin{claim}
	\label{cl:127_nochords}
	There is no $bc$-cycle in $G$.
\end{claim}

\begin{proofclaim}
	Suppose, to the contrary, that $C=u_1u_2\dots u_\ell$ is a $bc$-cycle of length $\ell$ in $G$.
	By Claim~\ref{cl:127_chords}, we already have that $C$ is chordless.
	Clearly, $\ell$ is even and by Claim~\ref{cl:127_square}, $\ell \ge 6$.
	For every $i$, $1 \le i \le \ell$, 
	let $u_i'$ be the neighbor of $u_i$ such that $\pi(u_iu_i') = a$ (and thus all $u_i'$ are distinct).
	Since $C$ is chordless, no $u_i'$ is a vertex of $C$.	
	We consider three cases regarding $\ell$.

	\begin{itemize}
		\item[$(a)$] $\ell \equiv 0 \bmod{3}$. \quad
			Let $G' = G - E(C)$. By the minimality, there exists an $a$-induced coloring $\sigma_\pi^a$ of $G'$.
			To extend $\sigma_\pi^a$ to all edges of $G$, we only need to color the edges of $C$.
			Since $C$ is chordless, the only conflicts among its edges are those generated by $C$.
			Hence, every edge of $C$ has at least $3$ available colors, 
			by Theorem~\ref{thm:listcyc}, every cycle of length divisible by $3$ is $3$-choosable. 
			Thus, we can extend $\sigma_\pi^a$ to $G$.

		\item[$(b)$] $\ell \equiv 2 \bmod{3}$ (and so $\ell \ge 8$). \quad
			In this case, we perform two crossings at the same time, one with the path $P = u_\ell u_1 u_2 u_3 u_4 u_5$,
			and the other with the path $P' = u_4 u_5 u_6 u_7 u_8 u_j$, where $j = 1$ if $\ell = 8$, and $j = 9$ otherwise.
			Note that the properties for the lists of available colors guaranteed in Claim~\ref{cl:127_crossing} still hold.
			Let 
			$$
				G' = \big(G \setminus (E(C) \cup \{u_iu_i'\}_{1\leq i \leq 8})\big) \cup \set{u_1'u_3', u_2'u_4', u_5'u_7', u_6'u_8'}.
			$$ 
			By the minimality, there exists an $a$-induced coloring $\sigma_\pi^a$ of $G'$;
			thus $u_1'u_3'$, $u_2'u_4'$, $u_5'u_7'$, and $u_6'u_8'$ are colored with $0$ in $\sigma_\pi^a$.
			Without loss of generality, we may assume that the two $2$-colors incident with $u_5'$ are $1$ and $2$, 
			and that the two $2$-colors incident with $u_7'$ are $3$ and $4$. 
			Denote by $S$ the set containing the two $2$-colors incident with $u_6'$.
			Moreover, for every $i$, $1 \le i \le \ell-1$, let $L_i$ be the set of available $2$-colors for the edge $u_iu_{i+1}$ 
			(and the list for the edge $u_\ell u_1$ we denote by $L_\ell$).
			As in the previous case, we have $\abs{L_i} \geq 3$ for every $i$.
		
			We consider two possibilities regarding $S$.
			Suppose first that $3 \notin S$. 
			Then, $3 \in L_5$ and we color $u_5u_6$ with $3$. Since $u_7'$ is incident with an edge colored with $3$, the sizes of $L_6$ and $L_7$ do not decrease. 
			Now, consider the lists of available colors for the edges of $P$. 
			By Claim~\ref{cl:127_crossing}, we have three possibilities. 
			Suppose first that $|L_2| = 5$. In this case, we may color it last, since it will have at least one available color after all the edges at distance 
			$2$ on $C$ are colored, thus we may ignore it for now. 
			We color $u_4u_5$ (decreasing the size of $L_6$ by one) and $u_3u_4$ (decreasing the size of $L_1$ by one). 
			It remains to color the edges of the path $u_6u_7\dots u_\ell u_1u_2$, which is possible by Lemma~\ref{lem:path_choose}, 
			and finally coloring $L_2$, a contradiction.
			Suppose next that there exists $x \in L_{2}$ such that $|L_{1}\setminus x| \ge 3$ and $|L_{3} \setminus x| \ge 3$. 
			In this case, color $u_2u_3$ with $x$.				
			We obtain $\abs{L_1} \geq 3$, $\abs{L_3} \geq 2$, $\abs{L_4} \geq 1$, $\abs{L_6} \geq 3$, \dots, $\abs{L_{\ell-1}} \geq 3$, and $\abs{L_{\ell}} \geq 2$. 
			Color $u_4u_5$ and then $u_3u_4$. 
			It remains to color the path $u_6\dots u_\ell u_1u_2$ (of length $\ell - 4 = 1 \bmod 3$), 
			which can be done by Lemma~\ref{lem:path_choose}$(a)$, a contradiction.
			Finally, suppose that $\abs{L_2} \ge 4$ and $L_1 \cap L_3 = \emptyset$.
			In this case, we again color $u_4u_5$ (decreasing $|L_6|$ and $|L_2|$ by one) and $u_3u_4$ (only decrasing $|L_2|$ by one).
			The remaining edges are lying on a path which is colorable by Lemma~\ref{lem:path_choose}, a contradiction.

			Therefore, by symmetry, we may assume that $S = \{3,4\}$.
			Moreover, by the same reasoning, the set of $2$-colors incident with $u_2'$ is the same as the set of $2$-colors incident with $u_3'$. 
			This altogether means that $L_2$ and $L_6$ both have size at least $5$, and thus we can color the edges $u_2u_3$ and $u_6u_7$ last. 
			We first color $u_3u_4$, $u_4u_5$, and $u_5u_6$ (decreasing $|L_7|$ and $|L_1|$ by one), 
			and then color the path $u_7\dots u_\ell u_1u_2$ by Lemma~\ref{lem:path_choose}. 
			Finally, color $u_2u_3$ and $u_6u_7$, a contradiction.

		\item[$(c)$] $\ell \equiv 1 \bmod{3}$ (and so $\ell \ge 10$). \quad
			As in the previous case, we perform two crossings at the same time, one with the path $P = u_\ell u_1 u_2 u_3 u_4 u_5$,
			and the other with the path $P' = u_5 u_6 u_7 u_8 u_9 u_j$, where $j = 1$ if $\ell = 10$, and $j = 11$ otherwise.			
			Let 
			$$
				G' = \big( G \setminus (E(C) \cup \{u_iu_i'\}_{1\leq i \leq 4}  \cup \{u_iu_i'\}_{6\leq i \leq 9}) \big) \cup \{u_1'u_3', u_2'u_4', u_6'u_8', u_7'u_9'\}.
			$$ 
			By the minimality, there exists an $a$-induced coloring $\sigma_\pi^a$ of $G'$;
			thus $u_1'u_3'$, $u_2'u_4'$, $u_6'u_8'$, and $u_7'u_9'$ are colored with $0$ in $\sigma_\pi^a$.
			Without loss of generality, we may assume that the two $2$-colors incident with $u_6'$ are $1$ and $2$,
			and the two $2$-colors incident with $u_8'$ are $3$ and $4$. 
			Denote by $S$ the set containing the two $2$-colors incident with $u_7'$.
			Moreover, for every $i$, $1 \le i \le \ell-1$, let $L_i$ be the set of available $2$-colors for the edge $u_iu_{i+1}$ 
			(and the list for the edge $u_\ell u_1$ we denote by $L_\ell$).
			Again, we have that $\abs{L_i} \geq 3$ for every $i$.
			
			We consider two possibilities regarding $S$.
			Suppose first that $3 \notin S$.
			Then, $3 \in L_6$ and we color $u_6u_7$ with $3$. Since $u_8'$ is incident with an edge colored with $3$, the sizes of $L_7$ and $L_8$ do not decrease. 
			Now, consider the lists of available colors for the edges of $P$. 
			By Claim~\ref{cl:127_crossing}, we have three possibilities. 
			Suppose first that $|L_2| = 5$. In this case, we may color it last, since it will have at least one available color after all the edges at distance 
			$2$ on $C$ are colored, thus we may ignore it for now. 
			We first consecutively color the edges $u_5u_6$, $u_4u_5$, and $u_3u_4$ (each of them has at least one available color when being colored),
			by that, we decrease the sizes of $L_1$ and $L_7$ by at most $1$, 
			and hence we can color the edges of the path $u_7\dots u_\ell u_1u_2$ by Lemma~\ref{lem:path_choose}. Finally, we color $u_2u_3$, a contradiction.
			Suppose next that there exists $x \in L_{2}$ such that $|L_{1}\setminus x| \ge 3$ and $|L_{3} \setminus x| \ge 3$. 
			In this case, color $u_2u_3$ with $x$, and then consecutively $u_4u_5$, $u_5u_6$, and $u_3u_4$. 
			It remains to color the edges of the path $P'' = u_7\dots u_\ell u_1u_2$, where every edge has at least three available colors, 
			except for the edges $u_1u_2$, $u_\ell u_1$, and $u_7u8$, which have at least $2$.
			Since the length of $P''$ is $\ell - 5 = 2 \mod 3$, we can color its edges by Lemma~\ref{lem:path_choose}$(a)$, a contradiction.
			Finally, suppose that $\abs{L_2} \ge 4$ and $L_1 \cap L_3 = \emptyset$.
			In this case, we first color $u_5u_6$ (decreasing $|L_7|$ by one), $u_4u_5$, $u_3u_4$ (not decrasing $|L_1|$), and $u_2u_3$ (decreasing $|L_1|$ and $|L_\ell|$ by one).
			The remaining edges are lying on a path which is colorable by Lemma~\ref{lem:path_choose}$(a)$, a contradiction.
			
			Therefore, by symmetry, we may assume that $S = \{3,4\}$.
			Moreover, by the same reasoning, the set of $2$-colors incident with $u_2'$ is the same as the set of $2$-colors incident with $u_3'$. 
			This altogether means that $L_2$ and $L_7$ both have size at least $5$, and thus we can color the edges $u_2u_3$ and $u_7u_8$ last. 			
			Now, color $u_3u_4$, $u_4u_5$, $u_5u_6$, and $u_6u_7$ in this order (decreasing $|L_1|$ and $|L_8|$ by one). 
			and then color the path $u_7\dots u_\ell u_1u_2$ by Lemma~\ref{lem:path_choose}. 
			Finally, color $u_2u_3$ and $u_7u_8$, a contradiction.			
	\end{itemize}
\end{proofclaim}

Since $G$ is cubic by Claim~\ref{cl:127_min3}, the subgraph of $G$ induced by the edges colored $b$ or $c$ in $\pi$ is $2$-regular, 
meaning that there must be at least one $bc$-cycle in $G$, which is in contradiction with Claim~\ref{cl:127_nochords}.
This establishes Theorem~\ref{th:127}.
\end{proof}

%
%
% OPEN PROBLEMS
%
%
\section{Further Work}
\label{sec:conc}

Conjecture~\ref{conj:main} remains open, but our upper bounds are only by one $2$-color off.
Unfortunately, we were not able to apply the proving techniques, used to prove tight bounds for proper edge-coloring and strong edge-coloring of subcubic graphs, 
to the problems considered in this paper. 
Therefore, since solving Conjecture~\ref{conj:main} in the general setting seems to be challenging, 
we suggest in this section additional problems which arise naturally when dealing with the considered colorings.
All of them are supported with computational results on graphs of small orders.

We begin with a general conjecture for strong edge-coloring.
\begin{conjecture}
	Every bridgeless subcubic graph $G$, not isomorphic to the Wagner graph or the complete bipartite graph $K_{3,3}$
	with one edge subdivided, admits a $(2^9)$-packing edge-coloring.
\end{conjecture}
%Not sure what we wanted to say here. \noindent The conjecture directly implies that the Wagner graph is the only bridgeless subcubic graph of class I.

\medskip
We proceed with an overview of results in specific graph classes and list open problems for each of them.
For that, we follow the conjecture on strong edge-coloring of subcubic graphs proposed
by Faudree, Gy\'{a}rf\'{a}s, Schelp, and Tuza~\cite{FauGyaSchTuz90} in 1990.
\begin{conjecture}[Faudree, Gy\'{a}rf\'{a}s, Schelp \& Tuza~\cite{FauGyaSchTuz90}]
	\label{con:Fau}
	For every subcubic graph $G$ it holds:
	\begin{itemize}
		\item[$(1)$] $G$ admits a $(2^{10})$-packing edge-coloring;
		\item[$(2)$] If $G$ is bipartite, then it admits a $(2^{9})$-packing edge-coloring;
		\item[$(3)$] If $G$ is planar, then it admits a $(2^{9})$-packing edge-coloring;
		\item[$(4)$] If $G$ is bipartite and each edge is incident with a $2$-vertex, then it admits a $(2^{6})$-packing edge-coloring; 
		\item[$(5)$] If $G$ is bipartite of girth at least $6$, then it admits a $(2^{7})$-packing edge-coloring;
		\item[$(6)$] If $G$ is bipartite and has girth large enough, then it admits a $(2^{5})$-packing edge-coloring.
	\end{itemize}
\end{conjecture}
All the cases of the conjecture, except $(5)$, are already resolved, and we present the results in what follows.

\subsection{Planar graphs}

It was the well-known connection between edge-coloring of bridgeless cubic planar graphs and the Four Color Problem, established by Tait~\cite{Tai80},
which initiated the research in this area. 
By the Four Color Theorem, we thus have that every bridgeless cubic planar graph admits a $(1,1,1)$-edge-coloring.
The condition of being bridgeless is necessary, since already $K_4$ with one subdivided edge is in class II.
However, not all questions are resolved. 
The following conjecture of Albertson and Haas~\cite{AlbHaa96}, 
which is a special case of Seymour's conjecture~\cite{Sey81}, is still widely open.
\begin{conjecture}[Albertson \& Haas~\cite{AlbHaa96}]
	Every bridgeless subcubic planar graph with at least two vertices of degree $2$ 
	admits a $(1,1,1)$-packing edge-coloring.
\end{conjecture}

The number of required colors for strong edge-coloring of planar graphs is also determined. 
Just recently, Kostochka et al.~\cite{KosLiRukSanWanYu16} proved the following (and resolved the case~$(3)$ of Conjecture~\ref{con:Fau}).
\begin{theorem}[Kostochka et al.~\cite{KosLiRukSanWanYu16}]
	\label{thm:planar}
	Every subcubic planar graph admits a $(2^9)$-packing edge-coloring.
\end{theorem}
The upper bound is tight due to the $3$-prism, depicted in Fig.~\ref{fig:3prism}.
For now, this is the only known planar graph with maximum degree $3$ with strong chromatic index equal to $9$.\footnote{In the 
conference version of this paper~\cite{HocLajLuz20}, it was erroneously stated that there exists an infinite family of such graphs.}
\begin{figure}[htp!]
	$$
		\includegraphics{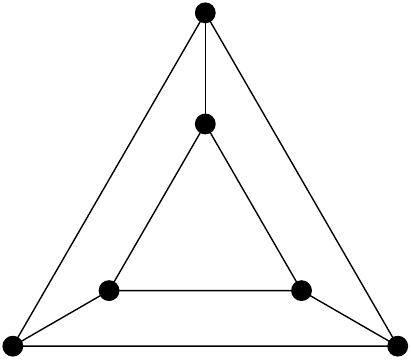}
	$$
	\caption{A cubic planar graph which needs nine colors for a strong edge-coloring.}
	\label{fig:3prism}
\end{figure}

On the other hand, there are no results for planar graphs on the colorings with one or two matchings.
We propose the following conjecture.
\begin{conjecture}
	Every subcubic planar graph admits a $(1,2^6)$-packing edge-coloring and a $(1,1,2^3)$-packing edge-coloring.
\end{conjecture}
The conjectured upper bound, if true, is tight and attained by an infinitely many bridgeless subcubic planar graphs
for both values. Indeed, in Figure~\ref{fig:plan12}, 
a planar bridgeless graph which does not admit a $(1,2^6)$-packing edge-coloring nor a a $(1,1,2^2)$-packing edge-coloring
and moreover, it can be appended to other subcubic graphs by the two $2$-vertices, 
hence an infinite family of bridgeless subcubic planar graphs not admitting such a coloring.
\begin{figure}[htp!]
	$$
		\includegraphics{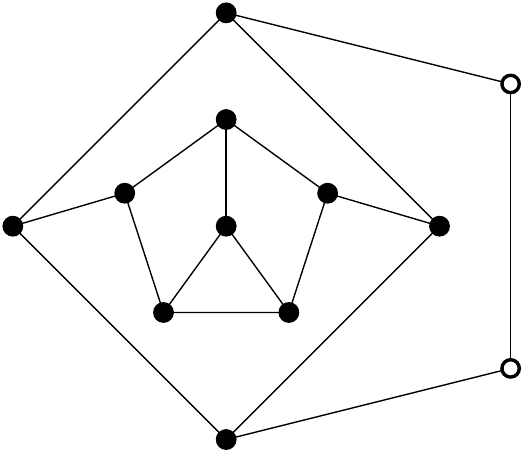}
	$$
	\caption{A subcubic planar graph which does not admit a $(1,2^6)$-packing edge-coloring nor a $(1,1,2^2)$-packing edge-coloring.}
	\label{fig:plan12}
\end{figure}
It also appears to be much more demanding as the result of Theorem~\ref{thm:planar}. 
Thus, also some partial results, with additional restrictions on the structure of planar graphs, 
might also be interesting, in order to understand the general problem better.

\subsection{Bipartite graphs}

In the class of bipartite graphs, the proper and the strong case of the colorings are long solved.
In 1916, K\"{o}nig~\cite{Kon16} proved that every bipartite graph is in class I, and in 1993,
Steger and Yu~\cite{SteYu93} established the following (and resolved the case~$(2)$ of Conjecture~\ref{con:Fau}).
\begin{theorem}[Steger \& Yu~\cite{SteYu93}]
	\label{thm:bipartite}
	Every subcubic bipartite graph admits a $(2^9)$-packing edge-coloring.
\end{theorem}
\noindent Again, the bound is tight; it is attained by, e.g., $K_{3,3}$.

Since all bipartite graphs are in class I, 
the results and conjectures for them apply also in the bipartite case.
It is known that as soon as we require some $2$-colors instead just $1$-colors, the problems become much harder.
E.g., a tight upper bound for a strong edge-coloring of bipartite graphs is still not known (c.f.~\cite{FauGyaSchTuz90,SteYu93}).
Therefore, the cases $(c)$ and $(d)$ of Conjecture~\ref{conj:main} may be considered just in the bipartite setting.
Moreover, we have an infinite number of graphs attaining the conjectured upper bounds also among bipartite graphs;
the bipartite graph with two $2$-vertices presented in Figure~\ref{fig:bip12} does not admit a $(1,2^6)$-packing edge-coloring 
nor a $(1,1,2^2)$-packing edge-coloring, and so an infite family of such graphs can again be constructed.
\begin{figure}[htp!]
	$$
		\includegraphics{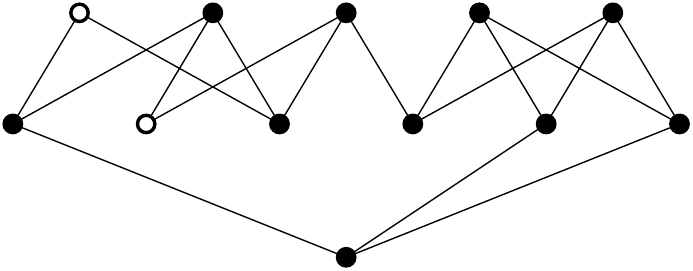}
	$$
	\caption{A subcubic bipartite graph which does not admit a $(1,2^6)$-packing edge-coloring nor a $(1,1,2^2)$-packing edge-coloring.}
	\label{fig:bip12}
\end{figure}

If we consider subcubic graphs with only edges of weight at most $5$, i.e., edges where at least one of the
end-vertices is of degree at most $2$, the number of required colors decreases substantially.
In particular, the case~$(4)$ of Conjecture~\ref{con:Fau} was resolved by Maydanskiy~\cite{May05}
and independently by Wu and Lin~\cite{WuLin08}.
\begin{theorem}[Maydanskiy~\cite{May05}, and Wu \& Lin~\cite{WuLin08}]
	\label{thm:bipartite}
	Every subcubic bipartite graph, in which each edge has weight at most $5$, admits a $(2^6)$-packing edge-coloring.
\end{theorem}

Clearly, an analogous question for coloring such graphs with two $1$-colors is if they admit
a $(1,1,2^2)$-packing edge-coloring. It is answered in affirmative~\cite{Sot19}. The bound is tight already in the class of trees.
On the other hand, we do not have the answer for the following.
\begin{question}
	\label{q:w5}
	Is it true that every subcubic bipartite graph, in which each edge has weight at most $5$, 
	admits a $(1,2^4)$-packing edge-coloring?
\end{question}
This bound is again attained in the class of trees.

\subsection{Graphs with big girth}

Similarly as the bipartiteness, having big girth does not really simplify edge-colorings in which some colors must be $2$-colors.
Even more, due to Kochol~\cite{Koc96} we know, there are graphs with arbitrarily large girth which are in class II.
Anyway, if the girth is infinite, i.e., we consider the trees, the following simple observation is immediate.
\begin{observation} 
	Every subcubic tree admits:
	\begin{itemize}
		\item[$(1)$] a $(1,1,1)$-packing edge-coloring;
		\item[$(2)$] a $(1,1,2^2)$-packing edge-coloring;
		\item[$(3)$] a $(1,2^4)$-packing edge-coloring;
		\item[$(4)$] a $(2^5)$-packing edge-coloring.
	\end{itemize}
\end{observation} 
\noindent The bounds are tight already if we just consider a neighborhood of one edge with both end-vertices of degree $3$.

In the case of strong edge-coloring, the case~$(6)$ of Conjecture~\ref{con:Fau} was also rejected just 
recently by Lu\v{z}ar, Ma\v{c}ajov\'{a}, \v{S}koviera, and Sot\'{a}k~\cite{LuzMacSkoSot20}, 
who proved that a cubic graph is a cover of the Petersen graph if and only if it admits a $(2^5)$-packing edge-coloring.

Before we consider the intermediate colorings, we first recall the result of Gastineau and Togni~\cite{GasTog19}.
\begin{proposition}[Gastineau \& Togni~\cite{GasTog19}]
	Every cubic graph admitting a $(1,1,2^2)$-packing edge-coloring is class I and has order divisible by four.
\end{proposition}
Hence, the analogue of the case~$(6)$ of Conjecture~\ref{con:Fau} when having two $1$-colors does not hold.
However, the following remains open.
\begin{question}
	\label{q:biggirth}
	Is it true that every subcubic bipartite graph with big enough girth
	admits a $(1,2^4)$-packing edge-coloring?
\end{question}

To conclude, we believe that studying properties of the considered edge-colorings will have impact to the initial 
problem of strong edge-coloring, which is in general case still widely open. 
Namely, the conjectured upper bound for graphs with maximum degree $\Delta$ is $1.25 \Delta^2$, 
while currently the best upper bound is due to Hurley, de Joannis de Verclos, and Kang~\cite{HurJoaKan20}, set at $1.772 \Delta^2$.

\bigskip
\paragraph{Acknowledgement} 

This research has been done in the scope of the bilateral project between France and Slovenia, BI-FR/19-20-PROTEUS-001.
The third author was partly supported by the Slovenian Research Agency Program P1--0383 and the project J1--1692.

\bibliographystyle{plain}
{\small
	\bibliography{MainBase}
}

\end{document}